\documentclass{elsarticle}

\pdfoutput=1

\usepackage{stmaryrd}
\usepackage{enumitem}
\usepackage{amsmath}
\usepackage{amssymb}
\usepackage{amsthm}
\usepackage{amsfonts}
\usepackage{centernot}
\usepackage{graphicx}
\usepackage[latin1]{inputenc}
\usepackage{subcaption}
\usepackage{xspace}
\usepackage{apxproof} % for moving things into the appendix

\DeclareMathOperator{\tlU}{\ensuremath\mathbf{U}}
\DeclareMathOperator{\tlX}{\ensuremath\mathbf{X}}
\DeclareMathOperator{\tlR}{\ensuremath\mathbf{R}}
\DeclareMathOperator{\tlF}{\ensuremath\mathbf{F}}
\DeclareMathOperator{\tlG}{\ensuremath\mathbf{G}}

\DeclareMathOperator{\tlWeakX}{\ensuremath\mathbf{\overline{X}}}

\newdefinition{definition}{Definition}
\newdefinition{example}{Example}
\newtheoremstyle{bfnote}%
	{}{}%
	{\itshape}{}%
	{\bfseries}{.}%
	{ }%
	{\thmname{#1}\thmnumber{ #2}\thmnote{ (#3)}}
\theoremstyle{bfnote}
\newtheorem{corollary}{Corollary}
\newtheorem{lemma}{Lemma}
\newtheorem{theorem}{Theorem}

\newtheoremrep{corollary}{Corollary}
\newtheoremrep{lemma}{Lemma}
\newtheoremrep{theorem}{Theorem}

\newcommand{\AP}{\mathcal{AP}} % Atomic propositions
\newcommand{\sat}{\models}
\newcommand{\ltlf}{LTL$_f$\xspace}
\newcommand{\den}[1]{\llbracket #1 \rrbracket}

\newcommand{\nats}{\mathbb{N}}

\newcommand{\logictrue}{\mathit{true}}
\newcommand{\logicfalse}{\mathit{false}}

\newcommand{\defeq}{\triangleq}

\begin{document}

\begin{frontmatter}
\title{A Tableau Construction for Finite Linear-Time Temporal Logic\tnoteref{t1}}
\tnotetext[t1]{
	Work supported by US National Science Foundation Grant CNS-1446365 and US Office of Naval Research Grant N000141712622.  
}
\author{Samuel Huang}
  \ead{srhuang@cs.umd.edu}
\author{Rance Cleaveland\corref{cor1}}
  \ead{rance@cs.umd.edu}
\address{Department of Computer Science, University of Maryland, College Park, Maryland 20740, USA}
\cortext[cor1]{Corresponding author}

\begin{abstract}
  This paper describes a method for converting formulas in finite propositional linear-time temporal logic (Finite LTL) into finite-state automata whose languages are the models of the given formula.  Finite LTL differs from traditional LTL in that formulas are interpreted with respect to finite, rather than infinite, sequences of states; this fact means that traditional finite-state automata, rather than $\omega$-automata such as those developed by B\"uchi and others, suffice for recognizing models of such formulas.  The approach considered is based on well-known tableau-construction techniques developed for LTL, which we adapt here for the setting of Finite LTL.  The resulting automata may be used as a basis for model checking, satisfiability testing, and model synthesis.
\end{abstract}

\end{frontmatter}

\section{Introduction}

Since its introduction into Computer Science by Amir Pnueli in a landmark paper~\cite{pnueli:sfcs77},
propositional linear temporal logic, or LTL, has played a prominent role as a specification formalism for discrete systems.  LTL includes constructs for describing how a system's state may change over time; this fact, coupled with its simplicity and the decidability properties it enjoys for both model checking and satisfiability checking, have made it an appealing framework for research into system verification and analysis.  The logic has also served as a springboard for the study of other temporal logics in computing, including ones that incorporate branching time~\cite{emerson1986sometimes}, fixpoints~\cite{kozen1983results}, real time~\cite{alur1993model} and probability~\cite{aziz1995usually}.

Traditional LTL formulas are interpreted with respect to infinite sequences of states, where each state assigns a truth value to the atomic propositions appearing in the formulas.  Such infinite sequences are intended to be viewed as runs of a system, with each state representing a snapshot of the system as it executes.  Traditional LTL has thus been used for defining properties of systems, such as servers and operating systems, that should not terminate and therefore whose executions are naturally viewed as infinite.  For systems whose executions are finite, alterations must be made, either in the approach for system modeling or in the variant or fragment of LTL used, to accommodate such behavior.  In the case of modeling, for example, a common practice is to extend finite executions ending in a terminal state with an infinite suffix of that state.  In the case of robotic-path planning, where LTL specifications are increasingly being used to describe temporal constraints that plans, which are finite-length, should satisfy,
fragments of LTL with a natural finite-sequence interpretation, such as the Generalized Reactivity(1) subset, are used for expressing these constraints~\cite{kress2009temporal}.

This semantic mismatch between traditional LTL and systems whose executions are finite has also led to the study of \emph{finite} variants of LTL, whose formulas are explicitly interpreted with respect to finite, rather than infinite, sequences of states.  Besides the examples cited above, domains as varied as general path planning~\cite{gerevini2009deterministic}, business-process specification~\cite{maggi2011runtime} and automated run-time monitoring~\cite{bauer2010comparing} have used such variants of LTL to precisely specify the desired behavior of systems.  This interest is motivation for the study of decision procedures, including model checking (does a system satisfy a formula?) and synthesis / satisfiability checking (is a formula satisfiable, and if so, can a satisfying sequence be computed automatically?), for such logics.

The purpose of this paper is to develop a construction for generating finite-state automata from a specific finite variant of LTL, which we call \emph{Finite LTL}, that accept exactly the models, or ``satisfying state sequences,'' of a given formula.  Such automata have natural applications in model checking, satisfiability checking, and synthesis of sequences satisfying a given Finite LTL formula.  Our approach relies on adapting the so-called 
\emph{tableau} construction of LTL~\cite{wolper1985tableau} to the setting of Finite LTL.  Tableau-based constructions for traditional LTL (e.g. \cite{courveur:fm1999}) have played a pivotal role in practical techniques for both model checking and satisfiability checking.   These constructions generally work by identifying states in the automaton under construction with certain sets of subformulas of the formula for which the automaton is being built.  In addition to the advantages that working with automata confers to the study of decision procedures of the logics in question, the tight linkage between formulas and states also enables certain practical advantages, including formula debugging, as the automaton allows one to interactively ``execute'' the formula.  It also enables other analyses to be undertaken, such as LTL query checking~\cite{chockler:hsvt2011,huang:avocs2017,huang2020arxiv-qc}, that can benefit from the connection between formulas and automata states.
To the best of our knowledge, no direct tableau construction has been given in the literature for a version of LTL interpreted over finite sequences, although another automaton construction has been defined for a related logic, \ltlf~\cite{degiacomo:ijcai2013}.  That technique relies on translation of \ltlf formulas into another logic, Linear Dynamic Logic interpreted over finite sequences, from which alternating finite automata~\cite{ashok1981alternation} are extracted that may then be converted into traditional finite automata.
  
Our construction exploits specific features of Finite LTL to simplify the traditional tableau constructions found for classical LTL.  In particular, it relies on semantics-preserving syntactic formula transformations to define the formulas associated with states, and the acceptance condition for the resulting automaton can be computed purely on the basis of the syntax of the formulas associated with each state.

The remainder of the paper is organized as follows. Section 2 reviews existing literature related to finite semantics of LTL and provides contrasting points and motivation for our own work.  Section 3 introduces the syntax and semantics of Finite LTL, discusses some nuances of the language, and shows that our logic is strictly more expressive than \ltlf, the finite version of LTL in~\cite{degiacomo:ijcai2013}.  Section 4 defines several syntactic normal forms into which Finite LTL formulas may be translated; these are used in the automaton construction presented in Section 5, which also briefly compares our approach to the construction in~\cite{degiacomo:ijcai2013}.  Section 6 describes our implementation of the tableau construction and the results of an empirical case study conducted using an existing benchmark from the LTL literature~\cite{duret.14.ijccbs}.  Finally, Section 7 concludes the paper.

\section{Background}
This section reviews existing work on the use of finite versions of LTL.  There have been a number of different LTL variants intended to bridge the gap between finite-length real-world data sequences and LTL \emph{vis \`a vis} infinite sequences in the formal-verification community.  Some have arisen with an intended application in mind (i.e. planning/robotics), while others are based on more foundational concerns.  As such, different works have yielded several variants of so-called ``finite LTL.'' We review some of these below.

De Giacomo and Vardi \cite{degiacomo:ijcai2013} provide a detailed complexity analysis of their finite LTL logic \ltlf.  They also devise a logic inspired by Propositional Dynamic Logic, Linear Dynamic Logic (LDL) over finite traces, into which \ltlf can be translated in linear time.  They show for both logics that determining satisfiability of a formula is PSPACE-complete and give a construction for converting LDL into alternating finite automata~\cite{ashok1981alternation}.  In another paper, De Giacomo et al.~\cite{degiacomo:aaai2014} focus on the interplay between finite and infinite LTL.  They particularly address some risks of directly transferring approaches from the infinite to the finite case.  They also formalize when an \ltlf formula is \emph{insensitive to infiniteness}, which holds when there exists an traditional LTL formula that is satisfied an infinite sequence exactly when the sequence is an  appropriately extended version of a finite sequence satisfying the \ltlf formula. It is also shown that this property is decidable for \ltlf formulas.

Li et al.~\cite{li:arxiv2014} apply transition systems to \ltlf satisfiability.  They present a recursive construction for a normal form of an \ltlf formula.  However, there are minor technical issues in their treatment of duality.  Ro{\c{s}}u \cite{rosu:rv2016} poses a sound and complete proof system for his version of finite-trace temporal logic.

Fionda and Greco \cite{fionda:aaai2016} investigate the complexity of satisfiability for restricted fragments of LTL over finite semantics.  They also provide an implementation of a finite LTL reasoner thatutilizes their complexity analysis to identify when satisfiability for a specified formula is computable in polynomial time, and performs the computation in the case; otherwise they convert the input formula to a SAT representation and invoke a SAT solver.  Notably, their maximal fragment does not include the dual operator of \emph{Next} or any binary operator within a temporal modality such as \emph{Until}, and only permits negation to be applied to atomic propositions.

A recent paper by Li et al.~\cite{li:cav2019} addresses Mission-Time LTL (MLTL), an LTL-based logic with time intervals supported for the \emph{Until} and \emph{Release} operators.  They develop a satisfiability checking tool for MLTL by first using a novel transformation from MLTL to LTL/LTL$_{f}$ and then turning this resulting formula into a SAT instance.  The authors also observe that there is a need for more solvers of these related languages.  This observation is an additional motivation for the work in this paper.  
%Across the surveyed literature, if any construction for an automaton representing a finite semantic LTL formula was provided, the acceptance criterion was recursively defined based on reachability from the initial state.  In our case, we have also provided a purely syntactic method of determining state acceptance.

In this paper, we have adopted a version of Finite LTL that is similar to that as used in~\cite{degiacomo:aaai2014,degiacomo:ijcai2013}.  Our choice is based firstly on a desire to appeal to the mature and well-studied similarities this representation has to the case of standard (infinite) LTL, and secondly to facilitate supplementing the logic to support the task of query checking~\cite{chan:cav2000} for finite sets of finite data streams.

\section{Finite LTL}

This section introduces the specific syntax and semantics of Finite LTL and compares its expressiveness with \ltlf~\cite{degiacomo:ijcai2013}. In what follows, fix a (nonempty) finite set $\AP$ of atomic propositions.

\subsection{Syntax of Finite LTL}

\begin{definition}[Finite LTL Syntax]\label{def:ltl-syntax}
The set of Finite LTL formulas is defined by the following grammar, where $a \in \AP$.
$$
\phi ::=  a \mid \lnot \phi \mid \phi_{1} \land \phi_{2} \mid \tlX \phi \mid \phi_{1} \tlU \phi_{2}
$$
We call the operators $\lnot$ and $\land$ \emph{propositional} and $\tlX$ and $\tlU$ \emph{modal}.
We use $\Phi^{\AP}$ to refer to the set of all Finite LTL formulas and $\Gamma^{\AP} \subsetneq \Phi^{\AP}$ for the set of all \emph{propositional} Finite LTL formulas, i.e.\/ those containing no modal operators.  We often write $\Phi$ and $\Gamma$ instead of $\Phi^{\AP}$ and $\Gamma^{\AP}$ when $\AP$ is clear from context.
\end{definition}

Finite LTL formulas may be constructed from atomic propositions using the traditional propositional operators $\lnot$ and $\land$, as well as the modalities of ``next'' ($\tlX$) and ``until'' ($\tlU$).
We also use the following derived notations:
\begin{align*}
\logicfalse	
&\defeq a \land \lnot a \\
\logictrue 	
&\defeq \logicfalse\\
\phi_1 \lor \phi_2
&\defeq \lnot ((\lnot \phi_1) \land (\lnot \phi_2))\\
\phi_1 \tlR \phi_2
&\defeq \lnot ((\lnot \phi_1) \tlU (\lnot \phi_1))\\
\tlWeakX \phi
&\defeq \lnot \tlX (\lnot \phi)\\
\tlF \phi
&\defeq \logictrue \tlU \phi\\
\tlG \phi
&\defeq \lnot \tlF (\lnot \phi)
\end{align*}

\noindent
The constants $\logicfalse$ and $\logictrue$, and the operators $\land$ and $\lor$, and $\tlU$ and $\tlR$, are duals in the usual logical sense, with $\tlR$ sometimes referred to as the ``release'' operator.  We introduce $\tlWeakX$ (``weak next'') as the dual for $\tlX$.  That this operator is needed is due to the semantic interpretation of Finite LTL with respect to finite sequences, which means that, in contrast to regular LTL, $\tlX$ is not its own dual.  This point is elaborated on later.  Finally, the duals $\tlF$ and $\tlG$ capture the usual notions of ``eventually'' and ``always'', respectively.

\subsection{Semantics of Finite LTL}

The semantics of Finite LTL is formalized as relation $\pi \sat \phi$, where $\pi \in (2^{\AP})^*$ is a finite sequence whose elements are subsets of $\AP$.  Such a subset $A \subseteq \AP$ represents a state $\sigma_A \in \AP \rightarrow \{0,1\}$, or assignment of truth values to atomic propositions, in the usual fashion:  $\sigma_A(a) = 1$ if $a \in A$, and $\sigma_A(a) = 0$ if $a \not\in A$.  We first introduce some notation on finite sequences.

\begin{definition}[Finite-Sequence Notation]
Let $X$ be a set, with $X^*$ the set of finite sequences of elements of $X$.  Also assume that $\pi \in X^*$ has form $x_0\ldots x_{n-1}$ for some $n \in \nats = \{0, 1, \ldots\}$. We define the following.
\begin{enumerate}
\item $\varepsilon \in X^*$ is the \emph{empty sequence}.
\item $|\pi| = n$ is the \emph{length} of $\pi$.  Note that $| \varepsilon | = 0$.
\item For $i \in \nats$, $\pi_i = x_i \in X$ provided $i < |\pi| = n$, and is undefined otherwise.
\item For $i \in \nats$, the \emph{suffix}, $\pi(i)$, of $\pi$ beginning at $i$ is taken to be $\pi(i) = x_i \ldots x_{n-1} \in X^*$, provided $i \leq |\pi| = n$ and is undefined otherwise.  Note that $\pi(0) = \pi$ and that $\pi(|\pi|) = \varepsilon$.
\item If $x \in X$ and $\pi \in X^*$ then $x \cdot \pi \in X^*$ is the sequence such that $(x \cdot \pi)_0 = x$ and $(x \cdot \pi)(1) = \pi$.  We often omit the $\cdot$ and write $x\pi$ rather than $x \cdot \pi$.
\end{enumerate}
\end{definition}

\begin{definition}[Finite LTL Semantics]\label{finite-ltl-semantics}
Let $\phi$ be a Finite LTL formula, and let $\pi \in (2^\AP)^*$.  Then the satisfaction relation, $\pi \sat \phi$, for Finite LTL is defined inductively on the structure of $\phi$ as follows.
\begin{itemize}
\item $\pi \sat a$ iff $|\pi| \geq 1$ and $a \in \pi_0$
\item $\pi \sat \lnot \phi$ iff $\pi \not\sat \phi$
\item $\pi \sat \phi_1 \land \phi_2$ iff $\pi \sat \phi_1$ and $\pi \sat \phi_2$
\item $\pi \sat \tlX \phi$ iff $|\pi| \geq 1$ and $\pi(1) \sat \phi$
\item $\pi \sat \phi_1 \tlU \phi_2$ iff $\exists j \colon 0 \leq j \leq |\pi| \colon \pi(j) \sat \phi_2$ and $\forall k \colon 0 \leq k < j \colon \pi(k) \sat \phi_1$
\end{itemize}
We write $\den{\phi}$ for the set $\{ \pi \mid \pi \sat \phi \}$.  We also say that $\phi_1$ and $\phi_2$ are \emph{logically equivalent}, notation $\phi_1 \equiv \phi_2$, if $\den{\phi_1} = \den{\phi_2}$.
%We sometimes call $\phi_1$ \emph{at least as strong as} $\phi_2$ (and $\phi_2$ as \emph{at least as weak as} $\phi_1$) if $\den{\phi_1} \subseteq \den{\phi_2}$.
\end{definition}

Intuitively, $\pi$ can be seen as an execution sequence of a system, with $\pi_0$, if it exists, taken to be the current state and $\pi_i$ for $i > 0$, if it exists, referring to the state $i$ time steps in the future.  In this interpretation $\varepsilon$ can be seen as representing an execution of a parameterized system whose initial state has not yet been configured.  Then $\pi \sat \phi$ holds if the sequence $\pi$ satisfies $\phi$.  Formula $a \in \AP$ can only be satisfied by non-empty $\pi$, as the presence or absence of $a$ in the first state $\pi_0$ contained in $\pi$ is used to determine whether $a$ is true ($a \in \pi_0$) or not ($a \not\in \pi_0$).  Negation  and conjunction are defined as usual.  A sequence $\pi$ satisfies $\tlX$ iff it is non-empty (and thus has a notion of ``next'') and the suffix of $\pi$ beginning after $\pi_0$ satisfies $\phi$.  Finally, $\tlU$ captures a notion of \emph{until}:  $\pi$ satisfies $\phi_1 \tlU \phi_2$ when it has a suffix satisfying $\phi_2$ and every suffix of $\pi$ that strictly includes this suffix satisfies $\phi_1$.

\subsection{Properties of Finite LTL}

Despite the close similarity of Finite LTL and LTL, the former nevertheless possesses certain semantic subtleties that we address in this section.  Many of these aspects of the logic have to do with properties of $\varepsilon$, the empty sequence, as a potential model of formulas.  (Indeed, other finite versions of LTL explicitly exclude non-empty sequences as possible models.)  The inclusion of $\varepsilon$ as a possible model for formulas simplifies the tableau construction given later in this paper, however; indeed, the definition of the acceptance condition that we give demands it.  Accordingly, this section also shows how the possibly counter-intuitive features of Finite LTL can be addressed with proper encodings.

\paragraph{$\tlX$ is not self-dual} 
In traditional LTL the $\tlX$ operator is \emph{self-dual}.  That is,  for any $\phi$, $\tlX \phi$ and $\lnot \tlX (\lnot \phi)$ are logically equivalent.  This fact simplifies the treatment of notions such as positive normal form, since no new operator needs to be introduced for the dual of $\tlX$.

In Finite LTL $\tlX$ does not have this property.  To see why, consider the formula $\tlX \logictrue$.  If $\tlX$ were self-dual then we should have  that $\tlX \logictrue \equiv \lnot(\tlX \lnot \logictrue)$, i.e.\/ that $\den{\tlX \logictrue} = \den{\lnot \tlX (\lnot \logictrue)}$.  However this fact does not hold.  Consider $\den{\tlX \logictrue}$.  Based on the semantics of Finite LTL$, \pi \sat \tlX \logictrue$ iff $|\pi| \geq 1$ and $\pi(1) \sat \logictrue$.  Since any sequence satisfies $\logictrue$, it therefore follows that
$$
\den{\tlX \logictrue} = \{ \pi \in (2^\AP)^* \mid |\pi| \geq 1 \};
$$
note that $\varepsilon \not\in \den{\tlX \logictrue}$.  Now consider $\den{\lnot \tlX \lnot \logictrue}$.  From the semantics of Finite LTL one can see that $\den{\lnot \logictrue} = \emptyset = \den{\tlX \lnot \logictrue}$.  It then follows that
$$
\den{\lnot \tlX \lnot \logictrue} = (2^\AP)^*,
$$
and thus $\varepsilon \in \den{\lnot \tlX \lnot \logictrue}$.  Consequently, $\den{\tlX \logictrue} \neq \den{\lnot \tlX \lnot\logictrue}$, and $\tlX$ is not self-dual.

The existence of duals for logical operators is used extensively in the tableau construction, so for this reason we have introduced $\tlWeakX$ as the dual for $\tlX$.  Using the semantics of $\tlX$ and $\lnot$ it can be seen that $\pi \sat \tlWeakX \phi$ iff either $|\pi| = 0$ or $\pi(1) \sat \phi$.  Indeed, $\den{\tlWeakX \phi} = \den{\tlX \phi} \cup \{ \varepsilon \}$; we sometimes refer to $\tlWeakX$ as the \emph{weak next} operator for this reason.

\paragraph{Including / excluding $\varepsilon$}
The discussion about $\tlX$ and $\tlWeakX$ above leads to the following lemma.

\begin{lemma}[Empty-sequence Formula Satisfaction]\label{lem:empty-sat}
Let $\pi \in (2^\AP)^*$.
\begin{enumerate}
\item $\pi \neq \varepsilon$ iff $\pi \sat \tlX \logictrue$.
\item $\pi = \varepsilon$ iff $\pi \sat \tlWeakX \logicfalse$.
\end{enumerate}
\end{lemma}
\begin{proof}
Immediate from the semantics of $\tlX$, $\tlWeakX$.
\end{proof}

\noindent
This lemma suggests a way for including / excluding $\varepsilon$ as a model of a formula. 

\begin{corollary}\label{cor:strengthen-weaken}
Let $\pi \in (2^\AP)^*$ and $\phi \in \Phi^\AP$.  Then the following hold.
\begin{enumerate}
\item $\pi \sat \phi \land \tlX \logictrue$ iff $\pi \sat \phi$ and $\pi \neq \varepsilon$; and
\item $\pi \sat \phi \lor \tlWeakX \logicfalse$ iff $\pi \sat \phi$ or $\pi = \varepsilon$.
\end{enumerate}
\end{corollary}

\paragraph{Literals and $\varepsilon$}
A \emph{literal} is a formula that has form either $a$ or $\lnot a$ for $a \in \AP$.  A positive-normal-form result for a logic asserts that any formula can be converted into an equivalent one in which all negations appear only as literals.

The semantics of Finite LTL dictates that for $\pi \sat a$ to hold, where $a \in \AP$, $\pi$ must be non-empty.  Specifically, the semantics requires that $|\pi| \geq 1$ and $a \in \pi_0$.  Based on the semantics of $\lnot$, it therefore follows that $\pi \sat \lnot a$ iff \emph{either} $|\pi| = 0$ \emph{or} $a \not\in \pi_0$.  It follows that $\varepsilon \sat \lnot a$ for any $a \in \AP$.

This may seem objectionable at first glance, since $\lnot a$ can be seen as asserting that $a$ is false ``now'' (i.e. in the current state), and $\varepsilon$ has no current state.  Given the semantics of $\lnot$ in Finite LTL, however, the conclusion is unavoidable.  However, using Corollary~\ref{cor:strengthen-weaken} we can give a formula that captures what might be the desired meaning of $\lnot a$, namely, that a satisfying sequence must be non-empty.  Consider $(\lnot a) \land \tlX \logictrue$.  From the corollary, it follows that $\pi \sat (\lnot a) \land \tlX \logictrue$ iff $\pi$ is non-empty and $a \not\in \pi_0$.

The relationship between $\varepsilon$ and literals also influences the semantics of formulas involving temporal operators.  For example, consider $\tlF \lnot a$ for literal $\lnot a$, which intuitively asserts that $a$ is eventually false.  More formally, based on the definition of $\tlF$ in terms of $\tlU$ and the semantics of $\tlU$, it can be seen that $\pi \sat \tlF \lnot a$ iff there exists $i$ such that $0 \leq i \leq |\pi|$ and $\pi(i) \sat \lnot a$.  Since for any $\pi$, $\pi(|\pi|) = \varepsilon$, it therefore follows that $\pi(|\pi|) \sat \lnot a$ for any $\pi$, and thus that every $\pi$ satisfies $\pi \sat \tlF \lnot a$.  This can be seen as offending intuition.  However, Corollary~\ref{cor:strengthen-weaken} again offers a helpful encoding.  Consider the formula $\tlF((\lnot a) \land \tlX \logictrue)$.  It can be seen that $\pi \sat \tlF((\lnot a) \land \tlX \logictrue)$ iff there is an $i$ such that $0 \leq i < |\pi|$ and $\pi(i) \sat \lnot a$, meaning that there must exist an $i$ such that $a \not\in \pi_i$.

A similar observation highlights a subtlety in the formula $\tlG a$ when $a \in \AP$.  It can be seen that $\pi \sat \tlG a$ iff for all $i$ such that $0 \leq i \leq |\pi|$, $\pi(i) \sat a$.  Since $\pi(|\pi|) = \varepsilon$ and $\varepsilon \not\sat a$, it therefore follows that $\tlG a$ is unsatisfiable.  This also seems objectionable, although Corollary~\ref{cor:strengthen-weaken} again offers a workaround.  Consider $\tlG (a \lor \tlWeakX \logicfalse)$.  In this case $\pi(|\pi|) \sat a \lor \tlWeakX \logicfalse$, and for all $i$ such that $0 \leq i < |\pi|$, $\pi(i) \sat a$ iff $a \in \pi_i$.  This formula captures the intuition that for $\pi$ to satisfy $a$, $a$ must be satisfied in every subset of $\AP$ in $\pi$.

\paragraph{Propositional formulas}
We close this discussion of the properties of Finite LTL with a study of the semantics of propositional formulas (i.e.\/ those not involving any propositional operators) in Finite LTL.  Later in this paper we rely extensively on traditional identities of propositional formulas, including De Morgan's Laws and distributivity, and the associated normal forms --- positive normal form and disjunctive normal form in particular --- that they enable.  In what follows we show that for the set $\Gamma^\AP$ of propositional formulas in Finite LTL, logical equivalence coincides with traditional propositional equivalence.  The only subtlety in establishing this claim has to do with the fact that in Finite LTL, $\varepsilon$ is allowed as a potential model.

We begin by recalling the traditional semantics of propositional formulas.

\begin{definition}[Semantics for Finite LTL Propositional Subset]\label{def:prop-sem}
Given a (finite, non-empty) set $\AP$ of atomic propositions, the \emph{propositional semantics} of formulas in  $\Gamma^\AP$ is given as a relation $\sat_p \;\subseteq 2^\AP \times \Gamma^\AP$ defined as follows.
\begin{enumerate}
\item $A \sat_p a$, where $a \in \AP$, iff $a \in A$.
\item $A \sat_p \lnot\gamma$ iff $A \not\sat_p \gamma$.
\item $A \sat_p \gamma_1 \land \gamma_2$ iff $A \sat_p \gamma_1$ and $A \sat_p \gamma_2$.
\end{enumerate}
We write $\den{\gamma}_p$ for $\{A \subseteq \AP \mid A \sat_p \gamma\}$ and $\gamma_1 \equiv_p \gamma_2$ when $\den{\gamma_1}_p = \den{\gamma_2}_p$.
\end{definition}

\noindent
Our goal is to show that for any $\gamma_1, \gamma_2 \in \Gamma^\AP$, $\gamma_1 \equiv \gamma_2$ iff $\gamma_1 \equiv_p \gamma_2$:  in other words, logical equivalence of propositional formulas in Finite LTL coincides exactly with traditional propositional logical equivalence.  In traditional LTL, this fact follows immediately from the fact that for infinite sequence $\pi$, $\pi \sat \gamma$ iff $\pi_0 \sat_p \gamma$.  In the setting of Finite LTL we have a similar result for non-empty $\pi$, but care must be taken with $\varepsilon$.

\begin{lemma}[Non-empty Sequence Propositional Satisfaction]\label{lem:non-empty-semantics}
Let $\pi \in (2^\AP)^*$ be such that $|\pi| > 0$, and let $\gamma \in \Gamma^\AP$.  Then $\pi \sat \gamma$ iff $\pi_0 \sat_p \gamma$.
\end{lemma}

\begin{proof}
Follows by induction on the structure of $\gamma$.
\end{proof}

\noindent
The next lemma establishes a correspondence between $\varepsilon$ satisfying propositional Finite LTL formulas and the propositional semantics of such formulas.

\begin{lemmarep}[Empty Sequence Propositional Satisfaction]\label{lem:vareps}
Let $\gamma \in \Gamma^\AP$ be a propositional formula.  Then $\varepsilon \sat \gamma$ iff $\emptyset \sat_p \gamma$.
\end{lemmarep}
\begin{proofsketch}
Follows by induction on the structure of $\gamma$.  See appendix for details.
\end{proofsketch}
\begin{proof}
The result follows by structural induction on $\gamma$.  There are three cases to consider
\begin{enumerate}
\item
$\gamma = a$ for some $a \in \AP$.  In this case $\varepsilon \not\sat a$ and $\emptyset \not\sat_p a$, so the desired bi-implication follows.
\item
$\gamma = \lnot\gamma'$ for some $\gamma' \in \Gamma^\AP$.  In this case the induction hypothesis says that $\varepsilon \sat \gamma'$ iff $\emptyset \sat_p \gamma'$.  The reasoning proceeds as follows.
\begin{align*}
\varepsilon \sat \gamma
& \;\text{iff}\; \varepsilon \sat \lnot\gamma'
&& \gamma = \lnot\gamma'
\\
& \;\text{iff}\; \varepsilon \not\sat \gamma'
&& \text{Definition of $\sat$}
\\
& \;\text{iff}\; \emptyset \not\sat_p \gamma'
&& \text{Induction hypothesis}
\\
& \;\text{iff}\; \emptyset \sat_p \lnot\gamma'
&& \text{Definition of $\sat_p$}
\\
& \;\text{iff}\; \emptyset \sat_p \gamma
&& \gamma = \lnot\gamma'
\end{align*}

\item
$\gamma = \gamma_1 \land \gamma_2$ for some $\gamma_1, \gamma_2 \in \Gamma^\AP$.  In this case the induction hypothesis guarantees the result for $\gamma_1$ and $\gamma_2$.  We reason as follows.
\begin{align*}
\varepsilon \sat \gamma
& \;\text{iff}\; \varepsilon \sat \gamma_1 \land \gamma_2
&& \gamma = \gamma_1 \land \gamma_2
\\
& \;\text{iff}\; \varepsilon \sat \gamma_1 \textnormal{ and } \varepsilon \sat \gamma_2
&& \text{Definition of $\sat$}
\\
& \;\text{iff}\; \emptyset \sat_p \gamma_1 \textnormal{ and } \emptyset \sat_p \gamma_2
&& \text{Induction hypothesis (twice)}
\\
& \;\text{iff}\; \emptyset \sat_p \gamma_1 \land \gamma_2
&& \text{Definition of $\sat_p$}
\\
& \;\text{iff}\; \emptyset \sat_p \gamma
&& \gamma =  \gamma_1 \land \gamma_2
\end{align*}\qedhere
\end{enumerate}
\end{proof}

\noindent
We can now state the main result of this section.

\begin{theoremrep}[Propositional / Finite LTL Semantic Correspondence]\label{thm:prop-logical-equivalence}
Let $\gamma_1, \gamma_2 \in \Gamma^\AP$.  Then $\gamma_1 \equiv \gamma_2$ iff $\gamma_1 \equiv_p \gamma_2$.
\end{theoremrep}
\begin{proofsketch}
Follows from Lemmas~\ref{lem:non-empty-semantics} and~\ref{lem:vareps}.  See appendix for details.
\end{proofsketch}
\begin{proof}
We break the proof into two pieces.
\begin{enumerate}
\item
Assume that $\gamma_1 \equiv \gamma_2$; we must show that $\gamma_1 \equiv_p \gamma_2$, i.e.\/ that for any $A \subseteq \AP, A \sat_p \gamma_1$ iff $A \sat_p \gamma_2$.  We reason as follows.
\begin{align*}
A \sat_p \gamma_1
& \;\text{iff}\; \pi \sat \gamma_1 \text{ all $\pi$ such that $|\pi| > 0$ and $\pi_0 = A$}
&& \text{Lemma~\ref{lem:non-empty-semantics}}
\\
& \;\text{iff}\; \pi \sat \gamma_2 \text{ all $\pi$ such that $|\pi| > 0$ and $\pi_0 = A$}
&& \gamma_1 \equiv \gamma_2
\\
& \;\text{iff}\; A \sat_p \gamma_2
&& \text{Lemma~\ref{lem:non-empty-semantics}}
\end{align*}

\item
Assume that $\gamma_1 \equiv_p \gamma_2$; we must show that $\gamma_1 \equiv \gamma_2$, i.e.\/ that for any $\pi \in (2^\AP)^*, \pi \sat \gamma_1$ iff $\pi \sat \gamma_2$.  So fix $\pi \in (2^\AP)^*$; we first consider the case when $|\pi| > 0$.
\begin{align*}
\pi \sat \gamma_1
& \;\text{iff}\; \pi_0 \sat_p \gamma_1
&& \text{Lemma~\ref{lem:non-empty-semantics}}
\\
& \;\text{iff}\; \pi_0 \sat_p \gamma_2 
&& \gamma_1 \equiv_p \gamma_2
\\
& \;\text{iff}\; \pi \sat \gamma_2
&& \text{Lemma~\ref{lem:non-empty-semantics}}
\end{align*}

We now consider the case when $|\pi| = 0$, meaning $\pi = \varepsilon$.
\begin{align*}
\varepsilon \sat \gamma_1
& \;\text{iff}\; \emptyset \sat_p \gamma_1
&& \text{Lemma~\ref{lem:vareps}}
\\
& \;\text{iff}\; \emptyset \sat_p \gamma_2 
&& \gamma_1 \equiv_p \gamma_2
\\
& \;\text{iff}\; \varepsilon \sat \gamma_2
&& \text{Lemma~\ref{lem:vareps}}
\end{align*}
\end{enumerate}
\end{proof}

\noindent Because of this theorem, propositional formulas in Finite LTL enjoy the usual properties of propositional logic.  In particular, in a logic extended with $\lor$ formulas can be converted into positive normal form, and disjunctive normal form, while preserving their semantics, including with respect to $\varepsilon$. 

\subsection{Finite LTL and \ltlf}\label{subsec:ltlf}

We close this section by considering the relative expressiveness of Finite LTL and the logic \ltlf given in~\cite{degiacomo:ijcai2013}.  In particular, we show that for every formula $\phi$ in \ltlf, there is a logically equivalent Finite LTL formula $\phi'$, but that the converse is not true:  there exists a Finite LTL formula that is not expressible in \ltlf.  Finite LTL is therefore strictly more expressive than \ltlf.

Syntactically, \ltlf as given in~\cite{degiacomo:ijcai2013} and Finite LTL are identical, modulo stylistic differences in the representation of the modalities.   Semantically, \ltlf formulas are interpreted with respect to pairs consisting of finite sequences in $(2^\AP)^*$ and \emph{positions}, or \emph{instants}, within the given sequence.  The following definitions are adapted from~\cite{degiacomo:ijcai2013}; the modifications are intended to clarify the treatment of $\varepsilon$ in the semantic account.

\begin{definition}[Instants of a Sequence]\label{def:instants}
Let $\pi \in (2^\AP)^*$.
\begin{enumerate}
\item
The \emph{instants}, $I(\pi) \subseteq \nats$, of $\pi$ are defined as $I(w) = \{i \in \nats \mid i < |\pi| \}$.  Note that $I(\varepsilon) = \emptyset$ and that $i \in I(\pi)$ iff $\pi_i \in 2^\AP$ is defined.
\item
If $\pi \neq \varepsilon$ then $\mathit{last}(\pi) = |\pi| - 1$ is the \emph{last position} in $\pi$.
\item
The set of \emph{\ltlf interpretations}, $S_{f}$, used to interpret \ltlf formulas is given by $S_{f} = \{ (\pi, i) \in (2^{\AP})^* \times \nats \mid i \in I(\pi)\}$.
\end{enumerate}
\end{definition}

\noindent
Note that since $I(\varepsilon) = \emptyset$, there can be no \ltlf interpretation of form $(\varepsilon, i)$, and that if $(\pi,i) \in S_{f}$ then $\pi_i$ is defined.  We now give the formal semantics of \ltlf formulas.

\begin{definition}[Semantics of \ltlf]\label{def:semantics-of-ltlf}
The semantics of \ltlf is given as a relation $\sat_f \;\subseteq S_f \times \Phi$ defined inductively as follows.
\begin{enumerate}
\item
$(\pi, i) \sat_f a$ iff $a \in \pi_i$.
\item
$(\pi, i) \sat_f \lnot\phi$ iff $(\pi, i) \not\sat_f \phi$.
\item
$(\pi, i) \sat_f \phi_1 \land \phi_2$ iff $(\pi, i) \sat_f \phi_1$ and $(\pi, i) \sat_f \phi_2$.
\item
$(\pi, i) \sat_f \tlX \phi$ iff $i < \mathit{last}(\pi)$ and $(\pi, i+1) \sat_f \phi$.
\item
$(\pi, i) \sat_f \phi_1 \tlU \phi_2$ iff for some $j$ such that $i \leq j \leq \mathit{last}(\pi)$, $(\pi, j) \sat_f \phi_2$ and for all $k$ such that $i \leq k < j, (\pi, k) \sat_f \phi_1$.
\end{enumerate}
We overload notation and write $\pi \sat_f \phi$ iff $(\pi, 0) \sat_f \phi$.  Note that if $\pi \sat_f \phi$ then $\pi \in (2^\AP)^+$, where $(2^\AP)^+ \subsetneq (2^\AP)^*$ is the set of non-empty sequences of $2^\AP$.  We write $\den{\phi}_f$ for $\{\pi \in (2^\AP)^+ \mid \pi \sat_f \phi\}$ and $\phi_1 \equiv_f \phi_2$ when $\den{\phi_1}_f = \den{\phi_2}_f$.
\end{definition}

We now prove that Finite LTL is at least as expressive as \ltlf.  The proof relies on the definition of a formula transformation, $T(\phi)$, whose purpose is to transform a \ltlf formula into a semantically equivalent formula in Finite LTL.

\begin{definition}[\ltlf to Finite LTL Transformation]
Transformation $T \in \Phi \rightarrow \Phi$ is defined inductively as follows.
\[
T(\phi) =
\begin{cases}
\phi							& \text{if $\phi \in \AP$}
\\
(\lnot T(\phi')) \land (\tlX\logictrue)	& \text{if $\phi = \lnot \phi'$}
\\
(T(\phi_1)) \land (T(\phi_2))		& \text{if $\phi = \phi_1 \land \phi_2$}
\\
\tlX (T (\phi'))					& \text{if $\phi = \tlX \phi'$}
\\
(T(\phi_1)) \tlU\, (T(\phi_2))		& \text{if $\phi = \phi_1 \tlU \phi_2$}
\end{cases}
\]
\end{definition}

The next lemma states a property of the semantics of \ltlf that is used in the proof of the theorem to follow.

\begin{lemma}[\ltlf semantic correspondence]\label{lem:ltlf-semantic-correspondence}
Let $\phi$ be a \ltlf formula.  Then for any $\pi \in (2^\AP)^+$ and $i \in I(\pi)$, $\pi, i \sat_f \phi$ iff $\pi(i) \sat_f \phi$.
\end{lemma}
\begin{proof}
Immediate from the definitions.
\end{proof}

We now have the following.
\begin{theoremrep}[Finite LTL Encodes \ltlf]\label{thm:finite-ltl-as-expressive-as-ltlf}
For any $\phi \in \Phi$, $\den{\phi}_f = \den{T(\phi)}$.
\end{theoremrep}

\begin{proofsketch}
Since $\den{\phi}_f \subseteq (2^\AP)^+$, $\den{T(\phi)} \subseteq (2^\AP)^*$ and $(2^\AP)^+ \subsetneq (2^\AP)^*$, it suffices to prove the following for all $\phi$.
\begin{enumerate}
\item %\label{stmt-1}
For all $\pi \in \den{T(\phi)},\pi \in (2^\AP)^+$
\item %\label{stmt-2}
For all $\pi \in (2^\AP)^+$, $\pi \sat T(\phi)$ iff $\pi \sat_f \phi$.
\end{enumerate}
The former statement guarantees that $\den{T(\phi)} \subseteq (2^\AP)^+$; the latter statement then ensures that $\den{T(\phi)} = \den{\phi}_f$.  Details may be found in the appendix.
\end{proofsketch}

\begin{proof}
Since $\den{\phi}_f \subseteq (2^\AP)^+$, $\den{T(\phi)} \subseteq (2^\AP)^*$ and $(2^\AP)^+ \subsetneq (2^\AP)^*$, it suffices to prove the following for all $\phi$.
\begin{enumerate}
\item\label{stmt-1}
For all $\pi \in \den{T(\phi)},\pi \in (2^\AP)^+$
\item\label{stmt-2}
For all $\pi \in (2^\AP)^+$, $\pi \sat T(\phi)$ iff $\pi \sat_f \phi$.
\end{enumerate}
The former statement guarantees that $\den{T(\phi)} \subseteq (2^\AP)^+$; the latter statement then ensures that $\den{T(\phi)} = \den{\phi}_f$.

The proof of Statement~(\ref{stmt-1}) proceeds by induction on the structure of $\phi$.  The induction hypothesis guarantees that this statement holds for all strict subformulas of $\phi$.  There are five cases to consider.
\begin{description}

\item[$\phi \in \AP$.]
Fix $\pi \in \den{T(\phi)}$.  We reason as follows.
\begin{align*}
\pi \in \den{T(\phi)}
& \;\text{iff}\; \pi \in \den{\phi}
&& \text{Definition of $T$}
\\
& \;\text{iff}\; \pi \sat \phi
&& \text{Definition of $\den{\phi}$}
\\
& \;\text{iff}\; |\pi| \geq 1 \;\text{and}\; \phi \in \pi_0
&& \text{Definition of $\sat$}
\\
& \;\text{implies}\; \pi \in (2^\AP)^+
&& \text{Definition of $(2^\AP)^+$}
\end{align*}

\item[$\phi = \lnot\phi'$.]
Fix $\pi \in \den{T(\phi)}$.  We reason as follows.
\begin{align*}
\pi \in \den{T(\phi)}
& \;\text{iff}\; \pi \in \den{(\lnot T(\phi')) \land (\tlX \logictrue)}
&& \text{Definition of $T$}
\\
& \;\text{iff}\; \pi \sat (\lnot T(\phi')) \land (\tlX \logictrue)
&& \text{Definition of $\den{\phi}$}
\\
& \;\text{implies}\; \pi \sat \tlX \logictrue
&& \text{Definition of $\sat$}
\\
& \;\text{iff}\; |\pi| \geq 1
&& \text{Definition of $\sat$}
\\
& \;\text{implies}\; \pi \in (2^\AP)^+
&& \text{Definition of $(2^\AP)^+$}
\end{align*}

\item[$\phi = \phi_1 \land \phi_2$.]
Fix $\pi \in \den{T(\phi)}$.  We reason as follows.
\begin{align*}
\pi \in \den{T(\phi)}
& \;\text{iff}\; \pi \in \den{(T(\phi_1)) \land (T(\phi_2))}
&& \text{Definition of $T$}
\\
& \;\text{iff}\; \pi \sat (T(\phi_1)) \land (T(\phi_2))
&& \text{Definition of $\den{\phi}$}
\\
& \;\text{iff}\; \pi \sat T(\phi_1) \;\text{and}\; \pi \sat T(\phi_2)
&& \text{Definition of $\sat$}
\\
& \;\text{implies}\; \pi \in (2^\AP)^+
&& \text{Induction hypothesis}
\end{align*}

\item[$\phi = \tlX \phi'$.]
Fix $\pi \in \den{T(\phi)}$.  We reason as follows.
\begin{align*}
\pi \in \den{T(\phi)}
& \;\text{iff}\; \pi \in \den{\tlX(T(\phi))}
&& \text{Definition of $T$}
\\
& \;\text{iff}\; \pi \sat \tlX(T(\phi))
&& \text{Definition of $\den{\phi}$}
\\
& \;\text{implies}\; |\pi| \geq 1
&& \text{Definition of $\sat$}
\\
& \;\text{implies}\; \pi \in (2^\AP)^+
&& \text{Definition of $(2^\AP)^+$}
\end{align*}

\item[$\phi = \phi_1 \tlU \phi_2$.]
Fix $\pi \in \den{T(\phi)}$.  We reason as follows.
\begin{align*}
& \pi \in \den{T(\phi)}
\\
& \text{iff}\; \pi \in \den{(T(\phi_1)) \tlU\, (T(\phi_2))}
&& \text{Definition of $T$}
\\
& \text{iff}\; \pi \sat (T(\phi_1)) \tlU\, (T(\phi_2))
&& \text{Definition of $\den{\phi}$}
\\
& \text{implies}\; \text{there exists $j \geq 0$ such that $\pi(j) \sat T(\phi_2)$}
&& \text{Definition of $\sat$}
\\
& \text{implies}\; \pi(j) \in (2^\AP)^+
&& \text{Induction hypothesis}
\\
& \text{implies}\; \pi \in (2^\AP)^+
&& \text{$|\pi| = j + |\pi(j)| \geq 1$}
\end{align*}

\end{description}

The proof of Statement~(\ref{stmt-2}) also proceeds by induction on the structure of $\phi$.  The induction hypothesis guarantees that the statement holds for all strict subformulas of $\phi$.  There are five cases to consider.
\begin{description}

\item[$\phi \in \AP$.]
Fix $\pi \in (2^\AP)^+$. We reason as follows.
\begin{align*}
\pi \sat T(\phi)
& \;\text{iff}\; \pi \sat \phi
&& \text{Definition of $T$}
\\
& \;\text{iff}\; |\pi| \geq 1 \;\text{and}\; \phi \in \pi_0
&& \text{Definition of $\sat$}
\\
& \;\text{iff}\; (\pi, 0) \sat_f \phi
&& \text{Definition of $(\pi, 0) \sat_f \phi$}
\\\
& \;\text{iff}\; \pi \sat_f \phi
&& \text{Definition of $\pi \sat_f \phi$}
\end{align*}

\item[$\phi = \lnot \phi'$.]
Fix $\pi \in (2^\AP)^+$. We reason as follows.
\begin{align*}
\pi \sat T(\phi)
& \;\text{iff}\; \pi \sat (\lnot T(\phi')) \land (\tlX \logictrue)
&& \text{Definition of $T$}
\\
& \;\text{iff}\; \pi \sat (\lnot T(\phi')) \;\text{and}\; \pi \sat (\tlX \logictrue)
&& \text{Definition of $\sat$}
\\
& \;\text{iff}\; \pi \not\sat T(\phi')
&& \text{Definition of $\sat$, $|\pi| \geq 1$}
\\
& \;\text{iff}\; \pi \not\sat_f \phi'
&& \text{Induction hypothesis}
\\
& \;\text{iff}\; \pi \sat_f \lnot\phi'
&& \text{Definition of $\sat_f$}
\\
& \;\text{iff}\; \pi \sat_f \phi
&& \text{$\phi = \lnot \phi'$}
\end{align*}

\item[$\phi = \phi_1 \land \phi_2$.]
Fix $\pi \in (2^\AP)^+$. We reason as follows.
\begin{align*}
\pi \sat T(\phi)
& \;\text{iff}\; \pi \sat (T(\phi_1)) \land (T(\phi_2))
&& \text{Definition of $T$}
\\
& \;\text{iff}\; \pi \sat T(\phi_1) \;\text{and}\; \pi \sat T(\phi_2)
&& \text{Definition of $\sat$}
\\
& \;\text{iff}\; \pi \sat_f \phi_1 \;\text{and}\; \pi \sat_f \phi_2
&& \text{Induction hypothesis}
\\
& \;\text{iff}\; \pi \sat_f \phi_1 \land \phi_2
&& \text{Definition of $\sat_f$}
\\
& \;\text{iff}\; \pi \sat_f \phi
&& \text{$\phi = \phi_1 \land \phi_2$}
\end{align*}

\item[$\phi = \tlX \phi'$.]
Fix $\pi \in (2^\AP)^+$. We reason as follows.
\begin{align*}
\pi \sat T(\phi)
& \;\text{iff}\; \pi \sat \tlX(T(\phi'))
&& \text{Definition of $T$}
\\
& \;\text{iff}\; |\pi| \geq 1 \;\text{and}\; \pi(1) \sat T(\phi')
&& \text{Definition of $\sat$}
\\
& \;\text{iff}\; \pi(1) \sat_f \phi'
&& \text{$\pi \in (2^{\AP})^+$, induction hypothesis}
\\
& \;\text{iff}\; \pi, 1 \sat_f \phi'
&& \text{Lemma~\ref{lem:ltlf-semantic-correspondence}}
\\
& \;\text{iff}\; \pi, 0 \sat_f \tlX \phi'
&& \text{Definition of $\sat_f$}
\\
& \;\text{iff}\; \pi \sat_f \phi
&& \text{Definition of $\sat_f$, $\phi = \tlX \phi'$}
\end{align*}

\item[$\phi = \phi_1 \tlU \phi_2$.]
Fix $\pi \in (2^\AP)^+$. We reason as follows.
\begin{flalign*}
&\pi \sat T(\phi)
\\
& \text{iff}\; \pi \sat (T(\phi_1)) \tlU\, (T(\phi_2))
&& \text{Definition of $T$}
\\
& \text{iff}\;
\\
& \multispan2{\hspace{1ex} there exists $j \geq 0$ such that $\pi(j) \sat T(\phi_2)$\hfil}
\\
& \multispan2{\hspace{1ex} and for all $0 \leq k < j, \pi(k) \sat T(\phi_1)$\hfil}
& \text{Definition of $\sat$}
\\
& \text{iff}\;
\\
& \multispan2{\hspace{1ex} there exists $j \geq 0$ such that $\pi(j) \sat_f \phi_2$\hfil}
\\
& \multispan2{\hspace{1ex} and for all $0 \leq k < j, \pi(k) \sat_f \phi_1$\hfil}
& \text{Induction hypothesis}
\\
& \text{iff}\;
\\
& \multispan2{\hspace{1ex} there exists $j \geq 0$ such that $\pi,j \sat_f \phi_2$\hfil}
\\
& \multispan2{\hspace{1ex} and for all $0 \leq k < j, \pi,k \sat_f \phi_1$\hfil}
& \text{Lemma~\ref{lem:ltlf-semantic-correspondence}}
\\
& \text{iff}\; \pi, 0 \sat_f \phi_1 \tlU \phi_2
&& \text{Definition of $\sat_f$}
\\
& \text{iff}\; \pi \sat_f \phi
&& \text{Definition of $\sat_f$, $\phi = \phi_1 \tlU \phi_2$}
\end{flalign*}
\end{description}
\end{proof}

We close this section by establishing that the Finite LTL formula $\lnot\tlX \logictrue$ is not expressible in \ltlf.  This fact implies that Finite LTL is strictly more expressive than \ltlf.

\begin{theorem}
There exists no \ltlf formula $\phi$ such that $T(\phi) \equiv \lnot\tlX \logictrue$.
\end{theorem}

\begin{proof}
Immediate from the fact $\den{\lnot\tlX \logictrue} = \{\varepsilon\}$ and Theorem~\ref{thm:finite-ltl-as-expressive-as-ltlf}, which implies that for any \ltlf formula $\phi$, $T(\phi) \subseteq (2^\AP)^+$ and thus $\varepsilon \not\in \den{T(\phi)}$.
\qedhere
\end{proof}

%\vspace{-1em}
\section{Normal Forms for Finite LTL}

The purpose of this paper is to define a construction for converting formulas in Finite LTL into non-deterministic finite automata (NFAs) with the property that the language of the NFA for a formula consists exactly of the finite sequences that satisfy the formula.  Such automata have many uses:  they provide a basis for model checking against Finite LTL specifications and for checking satisfiability of Finite LTL formulas.  The approach is adapted from the well-known tableau construction~\cite{vardi:automata1986} for LTL.  Our presentation relies on showing how Finite LTL formulas may be converted into logically equivalent formulas in a specific \emph{normal form}; this normal form will then be used in the construction given in the next section.

\subsection{Extended Finite LTL and Positive Normal Form}

Our construction works with Finite LTL formulas in \emph{positive normal norm} (PNF), in which negation is constrained to be applied to atomic propositions.  The PNF formulas in Finite LTL as given in Definition~\ref{def:ltl-syntax} are not as expressive as full Finite LTL; there are formulas $\phi$ in Finite LTL such that $\phi \not\equiv \phi'$ for any PNF $\phi'$ in Finite LTL.  However, if we extend Finite LTL by including duals of all operators in Finite LTL, we can obtain a logic whose formulas are as expressive as those in Finite LTL.

\begin{definition}[Extended Finite LTL Syntax]\label{def:efltl-syntax}
The set of Extended Finite LTL formulas is given by the following grammar, where $a \in \AP$.
$$
\phi ::= a \mid \lnot \phi \mid \phi_1 \land \phi_2 \mid \tlX \phi \mid \phi_1 \tlU \phi_2
\mid \phi_1 \lor \phi_2
\mid \tlWeakX \phi
\mid \phi_1 \tlR \phi_2
$$
We use $\Phi^\AP_e$ to refer to the set of all Extended Finite LTL formulas, and $\Gamma^\AP_e$ for the set of propositional Extended Finite LTL formulas (i.e.\/ formulas that do not include any use of $\tlX, \tlU, \tlWeakX$ or $\tlR$).
\end{definition}

\noindent
Extended Finite LTL extends Finite LTL by including the duals of $\land$, $\tlX$ and $\tlU$, namely, $\lor$, $\tlWeakX$ and $\tlR$, respectively.  Note that $\Phi^\AP \subsetneq \Phi^\AP_e$:  every Finite LTL formula is syntactically an Extended Finite LTL formula, but not vice versa.

The semantics of Extended Finite LTL is given as follows.

\begin{definition}[Extended Finite LTL Semantics]\label{def:efltl-semantics}
Let $\phi$ be an Extended Finite LTL formula, and let $\pi \in (2^\AP)^*$.  Then the semantics of Extended Finite LTL is given as a relation $\pi \sat_e \phi$ defined as follows.
\begin{itemize}
\item
$\pi \sat_e a$ iff $|\pi| \geq 1$ and $a \in \pi_0$.
\item
$\pi \sat_e \lnot \phi$ iff $\pi \not\sat_e \phi$.
\item
$\pi \sat_e \phi_1 \land \phi_2$ iff $\pi \sat_e \phi_1$ and $\pi \sat_e \phi_2$.
\item
$\pi \sat_e \tlX \phi$ iff $|\pi| \geq 1$ and $\pi(1) \sat_e \phi$.
\item
$\pi \sat_e \phi_1 \!\!\tlU\! \phi_2$ iff $\exists j \colon 0 \leq j \leq |\pi| \colon \pi(j) \sat_e \phi_2 \land \forall k \colon 0 \leq k < j \colon \pi(k) \sat_e \phi_1$.
\item
$\pi \sat_e \phi_1 \lor \phi_2$ iff either $\pi \sat_e \phi_1$ or $\pi \sat_e \phi_2$.
\item
$\pi \sat_e \tlWeakX \phi$ iff either $|\pi| = 0$ or $\pi(1) \sat_e \phi$.
\item
$\pi \sat_e \phi_1 \!\!\tlR\! \phi_2$ iff $\forall j \colon 0 \leq j \leq |\pi| \colon \pi(j) \sat_e \phi_2$ or $\exists k \colon 0 \leq k < j \colon \pi(k) \sat_e \phi_1$.
\end{itemize}
We define $\den{\phi}_e = \{ \pi \in (2^\AP)^* \mid \pi \sat_e \phi \} \lor \phi_1 \equiv_e \phi_2$ iff $\den{\phi_1}_e = \den{\phi_2}_e$.
\end{definition}

\noindent
The next lemmas establish relationships between Finite LTL and Extended Finite LTL.  The first 
shows that the semantics of Extended Finite LTL, when restricted to Finite LTL formulas, matches the semantics of Finite LTL.

\begin{lemma}[(Extended) Finite LTL Semantic Correspondence]\label{lem:extension}
Let $\phi$ be a formula in Finite LTL, and let $\pi \in (2^\AP)^*$.  Then $\pi \sat \phi$ iff $\pi \sat_e \phi$.
\end{lemma}
\begin{proof}
Immediate.
\end{proof}

\noindent
The next result establishes duality properties between the new operators in Extended Finite LTL and the existing ones in Finite LTL.

\begin{lemma}[Dualities in Extended Finite LTL]\label{lem:duality}
Let $\phi, \phi_1$ and $\phi_2$ be formulas in Extended Finite LTL, and let $\pi \in (2^\AP)^*$.  Then the following hold.
\begin{enumerate}
\item\label{duality-i}
$\pi \sat_e \phi_1 \lor \phi_2$ iff $\pi \sat_e \lnot((\lnot \phi_1) \land (\lnot \phi_2))$.
\item\label{duality-ii}
$\pi \sat_e \tlWeakX \phi$ iff $\pi \sat_e \lnot \tlX \lnot \phi$.
\item\label{duality-iii}
$\pi \sat_e \phi_1 \tlR \phi_2$ iff $\pi \sat_e \lnot((\lnot \phi_1) \tlU (\lnot \phi_2))$.
\end{enumerate}
\end{lemma}

\begin{proof}
Follows from the definition of $\sat_e$.
\end{proof}

\noindent
The next lemma establishes that although Extended Finite LTL includes more operators than Finite LTL, any Extended Finite LTL formula can be translated into a logically equivalent Finite LTL formula.  Thus, the two logics have the same expressive power.

\begin{lemma}[Co-expressiveness for (Extended) Finite LTL]\label{lem:fltl-efltl}
Let $\phi$ be an Extended Finite LTL formula.  Then there is a Finite LTL formula $\phi'$ such that $\den{\phi}_e = \den{\phi'}$.
\end{lemma}

\begin{proof}
Follows from Lemmas~\ref{lem:extension} and~\ref{lem:duality}.  The latter lemma in particular establishes that each non-Finite LTL operator in $\phi$ ($\lor$, $\tlWeakX$, $\tlR$) can be replaced by appropriately negated versions of its dual.  Specifically, $\phi_1 \lor \phi_2$ can be replaced by $\lnot ((\lnot \phi_1) \land (\lnot \phi_2))$, $\tlWeakX \phi'$ by $\lnot \tlX \lnot \phi$, and $\phi_1 \tlR \phi_2$ by $\lnot ((\lnot \phi_1) \tlU (\lnot \phi_2))$.
\end{proof}

\noindent
Although Extended Finite LTL does not enhance the expressive power of Finite LTL, it does enjoy a property that Finite LTL does not:  its formulas may be converted in \emph{positive normal form}.  This fact will be useful in defining the tableau construction; the relevant mathematical results are presented here.

\begin{definition}[Positive Normal Form (PNF)]\label{def:pnf}
The set of \emph{positive normal form} (PNF) formulas of Extended Finite LTL is defined inductively as follows.
\begin{itemize}
\item
If $a \in \AP$ then $a$ and $\lnot a$ are in positive normal form.
\item
If $\phi$ is in positive form then $\tlX \phi$ and $\tlWeakX \phi$ are in positive normal form.
\item
If $\phi_1$ and $\phi_2$ are in positive normal normal then $\phi_1 \land \phi_2$, $\phi_1 \lor \phi_2$, $\phi_1 \tlU \phi_2$ and $\phi_1 \tlR \phi_2$ are in positive normal form.
\end{itemize}
\end{definition}

\noindent
We now have the following.

\begin{lemma}[PNF and Extended Finite LTL]\label{lem:pnf}
Let $\phi \in \Phi^\AP_e$ be an Extended Finite LTL formula.  Then there is a $\phi' \in \Phi^\AP_e$ in PNF such that $\phi \equiv_e \phi'$.
\end{lemma}
\begin{proof}
Follows from the fact that $\lnot \lnot \phi \equiv_e \phi$ and the existence of dual operators in Extended Finite LTL, which enable identities such as $\lnot (\phi_1 \tlU \phi_2) \equiv_e (\lnot \phi_1) \tlR (\lnot \phi_2)$ to be used to ``drive negations'' down to atomic propositions.
\end{proof}

\subsection{Automaton Normal Form}

Propositional logic exhibits a number of logical equivalences that support the conversion of arbitrary formulas into various \emph{normal forms} that are then the basis for algorithmic analysis, including satisfiability checking.  Disjunctive Normal Form (DNF) is one such well-known normal form.  In this section we show how Extended LTL formulas in PNF can be converted into a normal form related to DNF, which we call \emph{Automaton Normal Form} (ANF); ANF will be a key vehicle for the automaton construction in the next section.  

We begin by reviewing the basics of DNF in the setting of the propositional fragment of Extended LTL.
We first lift the definitions of $\lor$ and $\land$ to finite sets of formulas in the usual manner.
\begin{definition}[Conjunction / Disjunction for Sets of Formulas]\label{def:indexed-prop-ops}
Let $P = \{ \phi_1, \ldots, \phi_n \}$, $n \geq 0$ be a finite set of Extended LTL formulas.  Then
$\bigwedge P$ and $\bigvee P$ are defined as follows.
\[
\begin{array}{c}
\bigwedge P =
	\begin{cases}
	\logictrue	& \text{if $n = 0$ (i.e.\/ $P = \emptyset$)}\\
	\phi_1	& \text{if $n = 1$ (i.e.\/ $P = \{\phi_1\}$)}\\
	(\bigwedge \{\phi_1, \ldots, \phi_{n-1}\}) \land \phi_n & \text{if $n \geq 2$}
	\end{cases}
\\[2em]
\bigvee P =
	\begin{cases}
	\logicfalse	&	\text{if $n = 0$ (i.e.\/ $P = \emptyset$)}\\
	\phi_1	& 	\text{if $n = 1$ (i.e.\/ $P = \{\phi_1\}$)}\\
	(\bigvee \{\phi_1, \ldots, \phi_{n-1}\}) \lor \phi_n		& \text{if $n \geq 2$}
	\end{cases}
\end{array}
\]
\end{definition}

\noindent
We now define disjunctive normal form as follows.

\begin{definition}[Disjunctive Normal Form (DNF)]\label{def:dnf}
\mbox{}
\begin{enumerate}
\item
A \emph{literal} is a formula of form $a$ or $\lnot a$ for some $a \in \AP$.
\item
A \emph{DNF clause} is a formula $C$ of form $\bigwedge \{\ell_1, \ldots, \ell_n\}$, $n \geq 0$, where each $\ell_i$ is a literal.
\item
A formula in $\Gamma^\AP_e$ is in \emph{disjunctive normal form (DNF)} if it has form $\bigvee \{C_1, \ldots C_k\}$, $k \geq 0$, where each $C_i$ is a DNF clause.
\end{enumerate}
\end{definition}

\noindent
The following is a well-known result in propositional logic that, due to Theorem~\ref{thm:prop-logical-equivalence}, is also applicable to the propositional fragment of Extended Finite LTL.

\begin{theorem}[DNF Conversion for Extended Finite LTL]\label{thm:dnf}
Let $\gamma \in \Gamma^\AP_e$.  Then there is a DNF formula $\gamma' \in \Gamma^\AP_e$ such that $\gamma \equiv_e \gamma'$.
\end{theorem}

\emph{Automaton normal form} (ANF) can be seen as an extension of DNF in which each clause is allowed to have a single subformula of form $\tlX \phi$ or $\tlWeakX \phi$, where $\phi$ is an formula in full Extended Finite LTL.  A clause in an ANF formula can be seen as defining whether or not a sequence $\pi$ satisfies the formula in terms of conditions that must hold on the first element of the sequence, if there is one, (the literals in the clause), and the rest of the sequence (the ``next-state'' formula in the clause).  This feature will be exploited in the automaton construction in the next section.  The formal definition of ANF is as follows.

\begin{definition}[Automaton Normal Form (ANF)]\label{def:anf}
\mbox{}
\begin{enumerate}
\item
An \emph{ANF clause} $C$ has form $(\bigwedge \{\ell_1, \ldots \ell_k\}) \land \mathbf{N} (\bigwedge\{\phi_1, \ldots, \phi_n\})$, where each $\ell_i$ is a literal, $\mathbf{N} \in \{\tlX, \tlWeakX\}$ and each $\phi_j \in \Phi^\AP_e$ is an arbitrary Extended Finite LTL formula.
\item
A formula in Extended Finite LTL is in \emph{automaton normal form (ANF)} iff it has form $\bigvee\{C_1, \ldots, C_k\}$, $k \geq 0$, where each $C_i$ is  an ANF clause.
\end{enumerate}
We often represent clauses as $(\bigwedge \mathcal{L}) \land \mathbf{N}(\bigwedge \mathcal{F})$, where $\mathcal{L}$ is a finite set of literals and $\mathcal{F}$ a finite set of Extended LTL formulas.  If $C = (\bigwedge \mathcal{L}) \land \mathbf{N}(\bigwedge \mathcal{F})$ we write
\begin{eqnarray*}
\textit{lits}(C)	&=& \mathcal{L} \\
\mathit{nf}(C) &=& \mathcal{F}
\end{eqnarray*}
for the set of literals and the set of ``next formulas'' following the next operator ($\tlX$ or $\tlWeakX$) in $C$.
\end{definition}

\noindent
The next lemma establishes a key feature of formulas in ANF \emph{vis \`a vis} the sequences  in $(2^\AP)^*$ that model it.

\begin{lemma}[Sequence Satisfaction and ANF]\label{lem:anf-satisfaction}
\mbox{}
\begin{enumerate}
\item\label{anf-satisfaction-i}
Let $C $ be an ANF clause. Then for any $\pi \in (2^\AP)^*$ such that $|\pi| > 0$, $\pi \sat_e C$ iff $\pi_0 \sat_p \bigwedge \mathit{lits}(C)$ and $\pi(1) \sat_e \bigwedge \mathit{nf}(C)$.
\item\label{anf-satisfaction-ii}
Let $\phi = \bigvee_i C_i$ be in ANF.  Then for every $\pi \in (2^\AP)^*$, $\pi \sat_e \phi$ iff $\pi \sat_e C_i$ for some $i$.
\end{enumerate}
\end{lemma}

\begin{proof}
For Part~\ref{anf-satisfaction-i}, let $\pi \in (2^\AP)^*$ be such that $|\pi| > 0$.  Also let $\mathcal{L} = \textit{lits}(C)$ and $\mathcal{F} = \mathit{nf}(C)$.  We reason as follows.
\begin{align*}
\pi \sat_e C
& \;\text{iff}\; \pi \sat_e (\bigwedge \mathcal{L}) \land \mathbf{N} (\bigwedge \mathcal{F})
&& \text{Definition~\ref{def:anf}}
\\
& \;\text{iff}\; \pi \sat_e \bigwedge \mathcal{L} \;\text{and}\;
   \pi \sat_e \mathbf{N} (\bigwedge \mathcal{F})
&& \text{Definition of $\sat_e$}
\\
& \;\text{iff}\; \pi_0 \sat_p \bigwedge \mathcal{L} \;\textnormal{and}\;
   \pi \sat_e \mathbf{N} (\bigwedge \mathcal{F})
&& \text{Lemma~\ref{lem:non-empty-semantics}, $\bigwedge \mathcal{L} \in \Gamma^\AP_e$}
\\
&  \;\text{iff}\; \pi_0 \sat_p \bigwedge \mathcal{L} \;\textnormal{and}\;
   \pi(1) \sat_e \bigwedge \mathcal{F}
&& \text{Definition of $\sat_e$}
\end{align*}

\noindent
Part~\ref{anf-satisfaction-ii} follows immediately from the semantics of $\bigvee$.
\end{proof}

\noindent
The importance of this lemma derives especially from its first statement.  This asserts that determining if an ANF clause is satisfied by a non-empty sequence can be broken down into a propositional determination about its initial state ($\pi_0$) and the literals in the clause, and a determination about the rest of the sequence ($\pi(1)$) and the ``next formulas'' of the clause.  This observation is central to the construction of automata from formulas that we give later.

In the rest of this section we will show that for any Extended Finite LTL formula $\phi$ there is a logically equivalent one in ANF.  We start by stating some logical identities that will be used later.
%The first lemma establishes distributivity laws that $\tlX$ and $\tlWeakX$ have with respect to $\land$ and $\lor$.

\begin{lemma}[Distributivity of $\tlX$, $\tlWeakX$]\label{lem:distributivity}
Let $\phi_1,\phi_2 \in \Phi^\AP_e$.
\begin{enumerate}
\item
$(\tlX \phi_1) \land (\tlX \phi_2) \equiv_e \tlX(\phi_1 \land \phi_2)$.
\item\label{iv}
$(\tlWeakX \phi_1) \land (\tlWeakX \phi_2) \equiv_e \tlWeakX(\phi_1 \land \phi_2)$.
\item\label{i}
$(\tlX \phi_1) \lor (\tlX \phi_2) \equiv_e \tlX(\phi_1 \lor \phi_2)$.
\item
$(\tlWeakX \phi_1) \lor (\tlWeakX \phi_2) \equiv_e \tlWeakX(\phi_1 \lor \phi_2)$.
\end{enumerate}
\end{lemma}
\begin{proof}
Immediate from the semantics of Extended Finite LTL.
\end{proof}

\noindent
The next lemma establishes that in a certain sense, $\tlX$ ``dominates'' $\tlWeakX$ in the context of conjunction.

\begin{lemma}[$\tlX$ Dominates $\tlWeakX$]\label{lem:domination}
The following holds for any Extended Finite LTL formulas $\phi_1, \phi_2$.
$$(\tlX \phi_1) \land (\tlWeakX \phi_2) \equiv_e \tlX(\phi_1 \land \phi_2)$$
\end{lemma}
\begin{proof}
Follows from the fact that if $\pi \sat (\tlX \phi_1) \land (\tlWeakX \phi_2)$ then $|\pi| > 0$.
\end{proof}

\noindent
The final lemma is key to our ANF transformation result.  It states that operators $\tlU$ and $\tlR$ may be rewritten using operators $\land$, $\lor$, $\tlX$ and $\tlWeakX$.

\begin{lemmarep}[Unrolling $\tlU$ and $\tlR$]\label{lem:unrolling}
The following holds for any Extended Finite LTL formulas $\phi_1$ and $\phi_2$.
\begin{enumerate}
\item\label{unrolling-i}
$\phi_1 \tlU \phi_2 \equiv_e \phi_2 \lor (\phi_1 \land \tlX(\phi_1 \tlU \phi_2))$.
\item\label{unrolling-ii}
$\phi_1 \tlR \phi_2 \equiv_e \phi_2 \land (\phi_1 \lor \tlWeakX(\phi_1 \tlR \phi_2))$.
\end{enumerate}
\end{lemmarep}
\begin{proofsketch}
Follows from the semantics of Extended Finite LTL.  Details may be found in the appendix.
\end{proofsketch}
\begin{proof}
We prove Statement~(\ref{unrolling-i}) by showing that
$
\den{\phi_1 \tlU \phi_2}_e = \den{\phi_2 \land (\phi_1 \lor \tlWeakX(\phi_1 \tlR \phi_2))}_e.
$
\begin{flalign*}
& \den{\phi_1 \tlU \phi_2}_e
\\
&{=}\; \{\pi \mid \pi \sat_e \phi_1 \tlU \phi_2\}
&& \text{Definition of $\den{-}_e$}
\\
&
\multispan3{${=}\; 
	\{
		\pi \mid 	\exists j \colon 0 \leq j \leq |\pi| \colon \pi(j) \sat_e \phi_2
	    			\;\text{and}\;
   				\forall i \colon 0 \leq i < j \colon \pi(i) \sat_e \phi_1
	\}$
\hfil}
\\
&
&& \text{Definition of $\sat_e$}
\\
&{=}\; \{\pi \mid \pi(0) \sat_e \phi_2\}
\\
\multispan4{$\qquad \cup\;
	\{
		\pi \mid 	\exists j \colon 1 \leq j \leq |\pi| \colon \pi(j) \sat_e \phi_2
	    			\;\text{and}\;
   				\forall i \colon 0 \leq i < j \colon \pi(i) \sat_e \phi_1
	\}$
\hfil}
\\
&
&& \text{Set theory}
\\
& 
\multispan3{${=}\;
	\den{\phi_2}_e
	\cup
	\{
		\pi \mid 	\exists j \colon 1 \leq j \leq |\pi| \colon \pi(j) \sat_e \phi_2
	    			\;\text{and}\;
   				\forall i \colon 0 \leq i < j \colon \pi(i) \sat_e \phi_1
	\}$
\hfil}
\\
&
&& \text{$\pi(0) = \pi$, Definition of $\den{-}_e$}
\\
&
\multispan3{${=}\;
	\den{\phi_2}_e \cup
	(
		\{\pi \mid \pi(0) \sat \phi_1\}
		\;\cap$
\hfil}
\\
&
\multispan3{\hfil 
		$\{\pi \mid \exists j \colon 1 \leq j \leq |\pi| \colon \pi(j) \sat_e \phi_2 \textnormal{ and }  \forall i \colon 1 \leq i < j \colon \pi(i) \sat_e \phi_1\}
	) $
}
\\
&
&& \text{Set theory}
\\
&
\multispan3{${=}\;
	\den{\phi_2}_e \cup
	(
		\den{\phi_1}_e
		\;\cap$
\hfil}
\\
&
\multispan3{\hfil 
		$\{\pi \mid \exists j \colon 1 \leq j \leq |\pi| \colon \pi(j) \sat_e \phi_2 \textnormal{ and }  \forall i \colon 1 \leq i < j \colon \pi(i) \sat_e \phi_1\}
	) $
}
\\
&
&& \text{$\pi(0) = \pi$, Definition of $\den{-}_e$}
\\
&
\multispan3{${=}\;
	\den{\phi_2}_e \cup
	(
		\den{\phi_1}_e
		\;\cap$
\hfil}
\\
&
\multispan3{\hfil 
		$\{\pi \mid \exists j' \colon 0 \leq j' \leq |\pi|-1 \colon \pi(j'+1) \sat_e \phi_2 \textnormal{ and }  \forall i' \colon 0 \leq i < j' \colon \pi(i'+1) \sat_e \phi_1\}
	) $
}
\\
&
&& j = j'+1, i=i'+1
\\
&
\multispan3{${=}\;
	\den{\phi_2}_e \cup
	(
		\den{\phi_1}_e
		\;\cap$
\hfil}
\\
&
\multispan3{\hfil 
		$\{\pi \mid \exists j' \colon 0 \leq j' \leq |\pi(1)| \colon \pi(1)(j') \sat_e \phi_2 \textnormal{ and }  \forall i' \colon 0 \leq i < j' \colon \pi(1)(i') \sat_e \phi_1\}
	) $
}
\\
&
&& |\pi(1)| = |\pi|-1, \pi(j'+1) = \pi(1)(j'), \pi(i'+1) = \pi(1)(i')
\\
&
\multispan3{${=}\;
	\den{\phi_2}_e \cup
	(
		\den{\phi_1}_e
		\cap
		\{\pi \mid \pi(1) \sat_e \phi_1 \tlU \phi_2\}
	) $
\hfil}
\\
&
&& \text{Definition of $\sat_e$}
\\
&
\multispan3{${=}\;
	\den{\phi_2}_e \cup
	(
		\den{\phi_1}_e
		\cap
		\{\pi \mid \pi \sat_e \tlX(\phi_1 \tlU \phi_2)\}
	) $
\hfil}
\\
&
&& \text{Definition of $\sat_e$}
\\
&
\multispan3{${=}\;
	\den{\phi_2}_e \cup
	(
		\den{\phi_1}_e
		\cap
		\den{\tlX(\phi_1 \tlU \phi_2)}_e\}
	) $
\hfil}
\\
&
&& \text{Definition of $\den{-}_e$}
\\
&
\multispan3{${=}\;
	\den{\phi_2 \lor (\phi_1 \land \tlX (\phi_1 \tlU \phi_2)}_e$
\hfil}
\\
&
&& \text{Definition of $\den{-}_e$}
%=
%& \den{\phi_2 \lor (\phi_1 \land \tlX (\phi_1 \tlU \phi_2))}_e
%& Semantics of $\land, \lor$
\end{flalign*}

To prove Statement~(\ref{unrolling-ii}), we can rely the duality of $\tlR$ and $\tlU$ and Statement~(\ref{unrolling-i}).  It suffices to show that $\den{\phi_1 \tlR \phi_2} = \den{\phi_2 \land (\phi_1 \lor \tlWeakX (\phi_1 \tlR \phi_2))}_e$.  We reason as follows.
\begin{flalign*}
& \den{\phi_1 \tlR \phi_2}_e
\\
& {=}\; \den{\lnot((\lnot \phi_1) \tlU (\lnot \phi_2))}
&& \text{Lemma~\ref{lem:duality}(\ref{duality-iii})}
\\
& {=}\; (2^{\AP})^* -  \den{(\lnot \phi_1) \tlU (\lnot \phi_2)}_e
&& \text{Definition of $\den{-}_e$}
\\
& {=}\; (2^\AP)^* - \den{(\lnot \phi_2) \lor ((\lnot \phi_1) \land \tlX  ((\lnot \phi_1) \tlU (\lnot \phi_2)))}_e
&& \text{Statement~(\ref{unrolling-i})}
\\
& {=}\; \den{\lnot\left((\lnot \phi_2) \lor ((\lnot \phi_1) \land \tlX ((\lnot \phi_1) \tlU (\lnot \phi_2)))\right)}_e
&& \text{Definition of $\den{-}_e$}
\\
& {=}\; \den{\phi_2 \land (\phi_1 \lor \lnot\tlX ((\lnot \phi_1) \tlU (\lnot \phi_2)))}_e
&& \text{Lemma~\ref{lem:duality}(\ref{duality-i})}
\\
& {=}\; \den{\phi_2 \land (\phi_1 \lor \tlWeakX \lnot((\lnot \phi_1) \tlU (\lnot \phi_2)))}_e
&& \text{Lemma~\ref{lem:duality}(\ref{duality-ii})}
\\
& {=}\;  \den{\phi_2 \land (\phi_1 \lor \tlWeakX (\phi_1 \tlR \phi_2))}_e
&& \text{Lemma~\ref{lem:duality}(\ref{duality-iii})}
\end{flalign*}
\end{proof}

\noindent
The remainder of this section will be devoted to proving the following theorem.

\begin{theorem}[Conversion to ANF]\label{thm:anf}
Let $\phi$ be an Extended Finite LTL formula in PNF.  Then there exists a transformation $\mathit{anf}$ such that $\mathit{anf}(\phi)$ is in ANF and the following hold.
\begin{enumerate}
\item
$\phi \equiv_e \mathit{anf}(\phi)$.
\item
Suppose $\mathit{anf}(\phi) = \bigvee C_i$. Then for each $C_i$ and each $\phi' \in \mathit{nf}(C_i)$, $\phi'$ is a subformula of $\phi$.
\end{enumerate}
\end{theorem}

\noindent
This theorem states that any PNF Extended LTL formula $\phi$ can be converted into ANF formula $\mathit{anf}(\phi)$, and in such away that each clause's ``next-state subformula'' consists of a conjunction of subformulas of $\phi$.  As any Extended LTL formula can be converted into PNF, this ensures that any Extended LTL formula can be converted into ANF.

To prove this theorem, we define several formula transformations that, when applied in sequence, yield a formula in ANF with the desired properties.  The first transformation ensures that all occurrences of $\tlU$ and $\tlR$ are \emph{guarded} in the resulting formula, in the following sense.

\begin{definition}[Guardedness]\label{def:guardedness}
Let $\phi$ be an Extended Finite LTL formula.
\begin{enumerate}
\item
Let $\phi'$ be a subformula of $\phi$.  Then $\phi'$ is \emph{guarded} in $\phi$ iff for every occurrence of $\phi'$ in $\phi$ is within an occurrence of a subformula of $\phi$ of form $\mathbf{N} \phi''$, where $\mathbf{N} \in \{\tlX, \tlWeakX\}$. 
\item
Formula $\phi$ is \emph{guarded} iff every subformula of $\phi$ of form $\phi_1 \tlU \phi_2$ or $\phi_1 \tlR \phi_2$ appears guarded in $\phi$.
\end{enumerate}
\end{definition}

\noindent
As an example of the above definition, consider formula $\phi = (a \tlU b) \land \tlX(a \tlU b)$.   This formula is not guarded, because the left-most occurrence of $(a \tlU b)$ does not appear within an occurrence of a subformula of form $\tlX \phi''$.  However, $\phi' = (b \lor (a \land \tlX (a \tlU b))) \land \tlX(a \tlU b)$ is guarded, and indeed $\phi' \equiv_e \phi$ due to Lemma~\ref{lem:unrolling}(\ref{unrolling-i}).

We now define a transformation $gt$ on formulas; the intent of this transformation is that $gt(\phi)$ is guarded, and $gt(\phi) \equiv_e \phi$.

\begin{definition}[Guardedness Transformation]\label{def:gt}
Extended Finite LTL formula transformation $gt$ is defined inductively as follows.
$$
gt(\phi) =
\left\{
\begin{array}{l@{\;\;\;}p{1.5in}}
a
& if $\phi = a$
\\
\lnot (gt(\phi'))
& if $\phi = \lnot \phi'$
\\
gt(\phi_1) \land gt(\phi_2)
& if $\phi = \phi_1 \land \phi_2$
\\
gt(\phi_1) \lor gt(\phi_2)
& if $\phi = \phi_1 \lor \phi_2$
\\
\phi
& if $\phi = \tlX \phi'$ or $\phi = \tlWeakX \phi'$
\\
gt(\phi_2) \lor (gt(\phi_1) \land \tlX \phi)
& if $\phi = \phi_1 \tlU \phi_2$
\\
gt(\phi_2) \land (gt(\phi_1) \lor \tlWeakX \phi)
& if $\phi = \phi_1 \tlR \phi_2$
\end{array}
\right.
$$
\end{definition}

\noindent
We have the following.

\begin{lemma}[Properties of $\mathit{gt}$]\label{lem:guardedness}
Let $\phi$ be an Extended Finite LTL formula.  Then:
\begin{enumerate}
\item\label{guardedness-i}
$gt(\phi)$ is guarded.
\item\label{guardedness-ii}
$gt(\phi) \equiv_e \phi$.
\item\label{guardedness-iii}
If $\phi$ is in PNF, then so is $gt(\phi)$.
\item\label{guardedness-iv}
Let $\mathbf{N} \phi'$ be a subformula of $gt(\phi)$, where $\mathbf{N} \in \{\tlX, \tlWeakX\}$.  Then $\phi'$ is a subformula of $\phi$.
\end{enumerate}
\end{lemma}
\begin{proof}
Immediate from the definition of $gt$ and Lemma~\ref{lem:unrolling}.
\end{proof}

\noindent
The next transformation we describe converts guarded Extended Finite LTL formulas into \emph{pseudo-ANF}.

\begin{definition}[Pseudo ANF]
\mbox{}
\begin{enumerate}
\item
An \emph{ANF pseudo-literal} has form $a, \lnot a$ or $\mathbf{N} \phi$, where $\mathbf{N} \in \{\tlX,\tlWeakX\}$ and $\phi \in \Phi^\AP_e$.  
\item
  %An \emph{ANF pseudo-clause} $C$ has form $\bigwedge \{\alpha_1, \ldots, \alpha_n\}$, $n \geq 0$, where each $\alpha_i$ is an ANF pseudo-literal, required to have exactly one such pseudo-literal. \marginpar{What does this last part mean?}
  An \emph{ANF pseudo-clause} $C$ has form $\bigwedge \{\alpha_1, \ldots, \alpha_n\}$, $n \geq 0$, where each $\alpha_i$ is an ANF pseudo-literal.
\item
A formula is in \emph{Pseudo-ANF} if it has form $\bigvee \{C_1, \ldots, C_n\}$, $n \geq 0$, where each $C_i$ is an ANF pseudo-clause.
\end{enumerate}
\end{definition}

\noindent
Note that every literal is also an ANF pseudo-literal.  An ANF pseudo-clause differs from an ANF clause in that the former may have multiple (or no) instances of pseudo-literals of form $\textbf{N}\phi$, while the latter is required to have exactly one, of form $\textbf{N} \bigwedge \mathcal{F}$.  We have the following.

\begin{lemma}[Conversion to Pseudo ANF]\label{lem:pseudo-anf}
Let $\phi$ be a guarded Extended Finite LTL formula in PNF.  Then there exists a formula $pa(\phi)$ such that:
\begin{enumerate}
\item\label{pseudo-anf-i}
$pa(\phi)$ is in Pseudo ANF.
\item\label{pseud-anf-ii}
$pa(\phi) \equiv_e \phi$.
\end{enumerate}
\end{lemma}

\begin{proof}
Transformation $pa$ is a version of the classical DNF transformation for propositional formulas in which the ANF pseudo-literals are treated as literals.
\end{proof}

\noindent
The final transformation, $an$, converts formulas in pseudo-ANF into semantically equivalent formulas in ANF.

\begin{definition}[Pseudo-ANF to ANF Conversion]\label{def:pseudo-clause-conversion}
\mbox{}
\begin{enumerate}
\item
Let $C = \bigwedge P$, where $P = \{\alpha_1, \ldots, \alpha_n\}$ is a set of ANF pseudo-literals and $n \geq 0$, be an ANF pseudo-clause.  Also let $L(P)$ be the literals in $P$ and $N(P) = P - L(P) = \{\textbf{N}_1\phi_1, \ldots, \textbf{N}_i\phi_i\}$ for some $0 \leq i \leq n$, each $\textbf{N}_i \in \{\tlX, \tlWeakX\}$,  be the non-literals in $P$.  Then $ct(C)$ is defined as follows.
$$
ct(C) =
\left\{
\begin{array}{lp{2in}}
(\bigwedge L(P)) \land \tlX(\bigwedge \{\phi_1, \ldots, \phi_i\})
& if $\mathbf{N}_j = \tlX$ for some $1 \leq j \leq i$
\\
(\bigwedge L(P)) \land \tlWeakX(\bigwedge \{\phi_1, \ldots, \phi_i\})
& otherwise
\end{array}
\right.
$$
\item
Let $\phi = \bigvee \{C_1, \ldots, C_n\}$, $n\geq 0$, be an Extended Finite LTL formula in Pseudo ANF.  Then transformation $an(\phi) = \bigvee \{ct(C_1), \ldots, ct(C_n)\}$.
\end{enumerate}
\end{definition}

\noindent
The next lemma and its corollary establish that $ct$ and $an$ convert pseudo-ANF clauses and formulas, respectively, into ANF clauses and formulas.

\begin{lemma}[Conversion from Pseudo ANF to ANF Clauses]\label{lem:anf-clause-conversion}
Let $C$ be a pseudo-ANF clause.  Then $ct(C)$ is an ANF clause, and $C \equiv_e ct(C)$.
\end{lemma}

\begin{proof}
Follows from Lemmas~\ref{lem:distributivity} and~\ref{lem:domination}.
\end{proof}

\begin{corollary}[Conversion from Pseudo ANF to ANF Formulas]\label{cor:anf-formula-conversion}
Let $\phi$ be a pseudo-ANF formula.  Then $an(\phi)$ is in ANF, and $\phi \equiv_e an(\phi)$.
\end{corollary}

\begin{proof}
Follows from Lemma~\ref{lem:anf-clause-conversion}.
\end{proof}

\noindent
We now have the machinery necessary to prove Theorem~\ref{thm:anf}.

\begin{proof}[Proof of Theorem~\ref{thm:anf}]
Let $\phi$ be an Extended Finite LTL formula in PNF.  We must show how to convert it into an ANF formula $\mathit{anf}(\phi) = \bigvee C_i$ such that $\phi \equiv_e \mathit{anf}(\phi)$ for each $C_i$ and each $\phi' \in \mathit{nf}(C_i)$, $\phi'$ is a subformula of $\phi$.

Define $\mathit{anf}(\phi) = an(pa(gt(\phi)))$; obviously $\mathit{anf}(\phi)$ is in ANF.  We now reason as follows.
$$
\begin{array}{rcl@{\;\;\;}p{2in}}
\phi
& \equiv_e
& gt(\phi)
& Lemma~\ref{lem:guardedness}; note $gt(\phi)$ is PNF
\\
& \equiv_e
& pa(gt(\phi))
& Lemma~\ref{lem:pseudo-anf}
\\
& \equiv_e
& an(pa(gt(\phi)))
& Corollary~\ref{cor:anf-formula-conversion}
\\
& \equiv_e
& \mathit{anf}(\phi)
& Definition of $\mathit{anf}$
\end{array}
$$
Thus $\mathit{anf}(\phi)$ is in ANF, and $\mathit{anf}(\phi) \equiv_e \phi$.

For the second part, we note that in the construction of $\mathit{anf}(\phi)$ we first compute $gt(\phi)$, which has the property that every subformula of form $\mathbf{N}\phi''$ is such that $\phi''$ is a subformula of $\phi$.  The definition of $pa$ guarantees that this property is preserved in $pa(gt(\phi))$.  Finally, the definition of $an$ ensures the desired result.
\end{proof}

\begin{example}[Conversion to ANF]
We close this section with an example showing how our conversion to ANF works.  Consider $\phi = a \tlU (b \tlR c)$; we show how to compute $an(pa(gt(\phi)))$.  Here is the result of $gt(\phi)$.
\begin{eqnarray*}
gt(\phi)	&=&	gt(a \tlU (b \tlR c))\\
			&=&	gt(b \tlR c) \lor (gt(a) \land \tlX \phi)\\
			&=&	(gt(c) \land (gt(b) \lor \tlWeakX (b \tlR c))) \lor (a \land \tlX \phi)\\
			&=&	(c \land (b \lor \tlWeakX (b \tlR c))) \lor (a \land \tlX \phi)
\end{eqnarray*}

\noindent
Note that this formula is guarded.  We now consider $pa(gt(\phi))$.
\begin{eqnarray*}
pa(gt(\phi)) 	&=& 	pa((c \land (b \lor \tlWeakX (b \tlR c))) \lor (a \land \tlX \phi))\\
				&=&	pa(((c \land b) \lor (c \land \tlWeakX (b \tlR c))) \lor (a \land \tlX \phi))\\
				&=&	\bigvee\{c \land b, c \land \tlWeakX (b \tlR c), a \land \tlX \phi\}
\end{eqnarray*}

\noindent
Note that two of the three clauses in $pa(gt(\phi))$ are already ANF clauses; the only that is not is $c \land b$.  This leads to the following.
\begin{eqnarray*}
\mathit{anf}(\phi)&=&	an(pa(gt(\phi)))\\
					&=&	an(\bigvee\{c \land b, c \land \tlWeakX (b \tlR c), a \land \tlX \phi\})\\
					&=&	\bigvee\{ct(c \land b), ct(c \land \tlWeakX (b \tlR c)), ct(a \land \tlX \phi)\}\\
					&=&	\bigvee\{c \land b \land \tlWeakX \logictrue, c \land \tlWeakX (b \tlR c), a \land \tlX \phi\}
\end{eqnarray*}

\noindent
Note that this formula is in ANF.  Also note that $ct(c \land b) = c \land b \land \tlWeakX \logictrue$ due to the fact that in pseudo-ANF clause $c \land b$ has no next-state pseudo-literals.  The definition of $ct$ ensures that $\tlWeakX \bigwedge \emptyset = \tlWeakX \logictrue$ is added to ensure that the result satisfies the syntactic requirements of being an ANF clause.

\end{example}

\section{A Tableau Construction for Finite LTL}\label{sec:tableau}

In this section we show how Finite LTL formulas may be converted into finite-state automata whose languages consist of exactly the sequences making the associated formula true.  Based on Lemma~\ref{lem:pnf} we know that any Finite LTL formula can be converted into an Extended Finite LTL formula in PNF, so in the sequel we show how to build finite automata from Extended Finite LTL formulas in PNF.  We begin by recalling the definitions of non-deterministic finite automata.

\begin{definition}[Non-deterministic Finite Automata (NFA)]\label{def:nfa}
\mbox{}
\begin{enumerate}
\item
A \emph{non-deterministic finite automaton} (NFA) is a tuple $(Q, \Sigma, q_I, \delta, F)$, where:
\begin{itemize}
\item
$Q$ is a finite set of \emph{states};
\item\label{nfa-i}
$\Sigma$ is a finite non-empty set of \emph{alphabet} symbols;
\item
$q_I \in Q$ is the \emph{start state};
\item
$\delta \subseteq Q \times \Sigma \times Q$ is the \emph{transition relation}; and
\item
$F \subseteq Q$ is the set of \emph{accepting states}.
\end{itemize}
\item\label{nfa-ii}
Let $M = (Q, \Sigma, q_I, \delta, F)$ be an NFA, let $q \in Q$, and let $w \in \Sigma^*$.  Then $q$ \emph{accepts $w$ in $M$} iff one of the following hold.
\begin{itemize}
\item
$w = \varepsilon$ and $q \in F$
\item
$w = \sigma w'$ for some $\sigma \in \Sigma$ and $w' \in \Sigma^*$, and there exists $(q,\sigma,q') \in \delta$ such that $q'$ accepts $w'$ in $M$.
\end{itemize}
\item\label{nfa-iii}
Let $M = (Q, \Sigma, q_I, \delta, F)$ be an NFA.  Then $L(M)$, the \emph{language} of $M$, is
$$
L(M) = \{ w \in \Sigma^* \mid q_I \textnormal{ accepts } w \textnormal{ in } M\}.
$$
\end{enumerate}
\end{definition}

\noindent
The following theorem formally states that for every Extended Finite LTL formula in PNF, there is an NFA whose language consists exactly of the sequences of states satisfying the formula.  The proof occurs later in this section.
% We now state the theorem we will prove in the rest of this section.

\begin{theorem}[Existence of NFAs for Extended LTL]\label{thm:tableau}
Let $\phi \in \Phi^\AP_e$ be in PNF.  Then there is an NFA $M_\phi$ such that $L(M_\phi) = \den{\phi}_e$.
\end{theorem}

\subsection{The Construction}

In this section we describe our approach for building the NFA $M_\phi$ mentioned in Theorem~\ref{thm:tableau} from Extended Finite LTL formula $\phi$ in PNF.
We have been referring to this process as a tableau construction, and indeed it makes essential use of identities, such as those in Lemmas~\ref{lem:duality}--\ref{lem:unrolling}, that also underpin classical tableau constructions.
However, because of our use of ANF we are able to avoid some complexities of other tableau constructions, such as the need for maximally consistent subsets of formulas as automaton states.

In what follows we use $S(\phi)$ to refer to the set of (not necessarily proper) subformulas of $\phi$.  States in $M_\phi$ will be associated with subsets of $S(\phi)$, and defining accepting states will require checking if $\varepsilon \sat_e \phi'$ for arbitrary $\phi' \in S(\phi)$.  The next lemma establishes that this latter check can be computed on the basis of the syntactic structure of $\phi'$.

\begin{lemma}[Empty-sequence Check]\label{lem:empty-check}
Let $\phi \in \Phi^\AP_e$ be in PNF.  Then $\varepsilon \sat_e \phi$ iff one of the following hold.
\begin{enumerate}
\item
$\phi = \lnot a$ for some $a \in \AP$
\item
$\phi = \phi_1 \land \phi_2$, $\varepsilon \sat_e \phi_1$, and $\varepsilon \sat_e \phi_2$
\item
$\phi = \phi_1 \tlU \phi_2$ and $\varepsilon \sat_e \phi_2$
\item
$\phi = \phi_1 \lor \phi_2$ and either $\varepsilon \sat_e \phi_1$ or $\varepsilon \sat_e \phi_2$
\item
$\phi = \tlWeakX \phi'$
\item
$\phi = \phi_1 \tlR \phi_2$ and $\varepsilon \sat_e \phi_2$
\end{enumerate}
\end{lemma}

\begin{proof}
Immediate from the definition of $\sat_e$.
\end{proof}

\noindent
We now formally define our tableau construction for $M_\phi$, an NFA that accepts exactly the finite sequences that satisfy $\phi$.
This is the key result of the paper.

\begin{definition}[The Tableau Construction for NFA $M_\phi$]\label{def:mphi}
Let $\phi \in \Phi^\AP_e$ be in PNF.  Then we define NFA $M_\phi = (Q_\phi, \Sigma_\AP, q_{I, \phi}, \delta_\phi, F_\phi)$ as follows.
\begin{itemize}
\item
$Q_\phi = 2^{S(\phi)}$
\item
$\Sigma_\AP = 2^\AP$
\item
$q_{I,\phi} = \{ \phi \}$  (Note that $ q_{I,\phi} = \{\phi\} \subseteq S(\phi)$, and thus $q_{I,\phi} \in Q_\phi$.)
\item
Let $q, q' \in Q_\phi$ (so $q, q' \subseteq S(\phi)$) and $A \in \Sigma_\AP$(so $A \subseteq \AP$).  Also let $\mathit{anf}(\bigwedge q) = \bigvee\{C_1, \ldots C_n\}$ be the ANF conversion of $\bigwedge q$.  Then $(q, A, q') \in \delta$ iff there exists $C_i$ such that:
\begin{itemize}
\item
$A \sat_p \bigwedge\textit{lits}(C_i)$; and
\item
$q' = \mathit{nf}(C_i)$.
\end{itemize}
\item
$F_\phi = \{q \in Q_\phi \mid \varepsilon \sat_e \bigwedge q \}$
\end{itemize}
\end{definition}

\noindent
Theorem~\ref{thm:tableau} states that the above construction is correct.  In the rest of this section, we will prove this claim.  We first establish the following useful lemma.

\begin{lemma}[Well-Formedness of $M_\phi$]\label{lem:well-formed}
Let $\phi \in \Phi^\AP_e$ be in PNF, and let $M_\phi = (Q_\phi, \Sigma_\AP, q_{I,\phi}, \delta_\phi, F_\phi)$. Fix arbitrary $q \in Q$, and let
$$
\mathit{anf}(\bigwedge q) = \bigvee C_i,
$$
Then for each $C_i$, $\mathit{nf}(C_i) \in Q_\phi$.
\end{lemma}

\begin{proof}
Follows from Lemma~\ref{lem:guardedness}(\ref{guardedness-iv}) and the fact that every subformula of every $\phi' \in \mathit{nf}(C_i)$ of form $\textbf{N}\phi''$, $\phi_1' \tlU \phi_2'$ or $\phi_1' \tlR \phi_2'$ is also a subformula of $\phi$.
\end{proof}

\noindent
This lemma in effect says that every clause $C_{i}$ occurring in $\mathit{anf}(\bigwedge q)$ (recall $q$ is a set of subformulas of $\phi)$ gives rise to transitions between states of $M_\phi$, because the ``next-state'' formulas in such a clause involve only subformulas of $\phi$.

\begin{proof}[Proof of Theorem~\ref{thm:tableau}]
We now prove Theorem~\ref{thm:tableau} as follows.  Let $\phi \in \Phi^\AP_e$ and $M_\phi = (Q_\phi, \Sigma_\AP, q_{I,\phi}, \delta_\phi, F_\phi)$.  We recall that $\Sigma_\AP = 2^\AP$, and thus $(\Sigma_\AP)^* = (2^\AP)^*$.  Consequently, the words accepted by $M_\phi$ come from the same set as the sequences to interpret Extended Finite LTL formulas.  To emphasize this connection, we use $A \in \Sigma_\AP$ and $\pi \in (\Sigma_\AP)^*$ in the following.  We will in fact prove a stronger result:  for every $\pi \in (\Sigma_\AP)^*$ and $q \in Q$, $q$ accepts $\pi$ in $M_\phi$ iff $\pi \sat_e \bigwedge q$.  The desired result then follows from the fact that this statement holds in particular for the start state, $q_{I,\phi}$, that $\bigwedge q_{I,\phi} = \phi$, and that as a result, $L(M_\phi) = \den{\phi}_e$.

The proof proceeds by induction on $\pi$.  For the base case, assume that $\pi = \varepsilon$ and fix $q \in Q$.  We reason as follows.
$$
\begin{array}{l@{\;\textnormal{iff}\;}l@{\;\;\;}p{2in}}
\textnormal{$q$ accepts $\varepsilon$ in $M$}
& q \in F_\phi
& Definition of acceptance
\\
& \varepsilon \sat_e \bigwedge q
& Definition of $F_\phi$
\end{array}
$$

In the induction case, assume $\pi = A \pi'$ for some $A \in \Sigma_\AP$ (so $A \subseteq \AP$) and $\pi' \in (\Sigma_\AP)^*$.  The induction hypothesis states for any $q' \in Q$, $q'$ accepts $\pi'$ in $M_\phi$ iff $\pi' \sat_e \bigwedge q'$ (recall $q' \subseteq S(\phi)$).  Now fix $q \in Q$; we must prove that $q$ accepts $\pi$ in $M_\phi$ iff $\pi \sat_e \bigwedge q$.  We reason as follows.
\begin{flalign*}
& \pi \sat_e \bigwedge q
\\
& \text{iff}\; A \pi' \sat_e \bigwedge q
&& \text{$\pi = A \pi'$}
\\
& \text{iff}\; A \pi' \sat_e \mathit{anf}(\bigwedge q)
&& \text{Theorem~\ref{thm:anf}}
\\
& \text{iff}\; A \pi' \sat_e \bigvee C_i
&& \text{$\mathit{anf}(\bigwedge q) = \bigvee C_i$ in ANF}
\\
& \text{iff}\; A \pi' \sat_e C_i \;\text{some $i$}
&& \text{Lemma~\ref{lem:anf-satisfaction}(\ref{anf-satisfaction-ii})}
\\
& \text{iff}\; A  \sat_p \bigwedge \mathcal{L} \textnormal{ and } \pi' \sat_e \bigwedge\mathcal{F}
&& \text{Lemma~\ref{lem:anf-satisfaction}(\ref{anf-satisfaction-i}),}
\\
&
&& \text{$\mathcal{L} = \textit{lits}(C_i)$, $\mathcal{F} = \mathit{nf}(C)$}
\\
& \text{iff}\; A  \sat_p \bigwedge \mathcal{L} \text{ and } \pi' \sat_e \bigwedge q' \text{ some } q' \in Q_\phi
&& \text{Lemma~\ref{lem:well-formed}}
\\
& \text{iff}\; (q, A, q') \in \delta_\phi \text{ and } \pi' \sat_e \bigwedge q' \text{ some } q' \in Q_\phi
&& \text{Definition of $\delta_\phi$}
\\
& \text{iff}\; (q, A, q') \in \delta_\phi \textnormal{ and } q' \textnormal{ accepts } w' \textnormal{ in } M
&& \text{Induction hypothesis}
\\
& \text{iff}\; q \textnormal{ accepts } w \text{ in } M
&& \text{Definition~\ref{def:nfa}(\ref{nfa-ii})}
\end{flalign*}
\end{proof}

%\vspace{-1em}
\subsection{Discussion of the Construction of $M_\phi$}\label{subsec:discussion}

We now discuss the $M_\phi$ construction, both from the standpoint of its complexity but also in terms of heuristics for improving the runtime of the construction as well as the size of the resulting NFAs.
% NFA-size performance of the construction in practice.
We also give a brief comparison with the NFA construction for \ltlf outlined in~\cite{degiacomo:ijcai2013}.

\paragraph{Size of $|M_\phi|$}

The key drivers for the size of $M_\phi$ are the sizes of its state space $Q_\phi$ and of its transition relation $\delta_\phi$.  The next theorem characterizes these.

\begin{theorem}[Bounds on Size of $M_\phi$]\label{thm:size}
Let $\phi \in \Phi^\AP_e$ be in PNF, and let $M_\phi = (Q_\phi, \Sigma_\AP, q_{I,\phi}, \delta_phi, F_\phi$).  Then we have the following.
\begin{enumerate}
\item\label{size-i}
$|Q_\phi| \leq 2^{|\phi|}$
\item\label{size-ii}
$|\delta_\phi| \leq 4^{|\phi|} \cdot 2^{|\AP|}$
\end{enumerate}
\end{theorem}
\begin{proof}
For the first statement, we note that there is a state in $M_\phi$ for each subset of $S(\phi)$, and that there are at most $2^{|\phi|}$ such subsets.  The second follows from the fact that each pair of states can have at most $2^{|\AP|}$ transitions between them.
\end{proof}

\noindent
It is worth noting that in the above result, the bound on the number of states is tight:  it is $2^{|\phi|}$, not e.g.\/ $2^{O(|\phi|)}$, which some tableau constructions for LTL yield.  Also note that if $\phi$ contains multiple instances of the same subformula, then $|S(\phi)| < |\phi|$; this explains the inequality in Statement (\ref{size-i}). 

\paragraph{Optimizing $M_\phi$}
The size results in Theorem~\ref{thm:size} are consistent with other tableau constructions; an automaton's size is in the worst case exponential in the size of the formula for which it is being constructed.
This worst-case behavior cannot be avoided in general, but it can often be mitigated heuristically for many formulas.
In what follows we consider several such heuristics.

\underline{On-the-fly Construction of $Q_\phi$.}
The construction in Definition~\ref{def:mphi} may be seen as pre-computing all possible states of $M_\phi$.
In practice many of these states are unreachable from the initial state; thus, adding them to $Q_\phi$ and then computing their outgoing transitions is unnecessary work.
One method for avoiding this work is to construct $Q_\phi$ in a demand-driven, or \emph{on-the-fly} manner.
Specifically, one starts with the state $q_{I,\phi}$ and adds this to $Q_\phi$.
Then one repeatedly does the following:  select a state $q$ in the current $Q_\phi$ whose transitions have not been computed, compute $q$'s transitions, adding states into $Q_\phi$ as needed so that each transition has a target in $Q_\phi$.
This process stops when transitions have been computed for all states in  $Q_\phi$.
The result of this strategy is that only states reachable from $q_{I,\phi}$ will be added into $Q_\phi$.

\underline{Symbolic Representation of Transitions.}
In Definition~\ref{def:mphi} transition  labels are represented concretely, as sets of atomic propositions.
One can instead allow transition labels that are \emph{symbolic}:  these labels have form $\gamma \in \Gamma^\AP_e$ for some propositional Extended LTL formula $\gamma$.
A transition labeled by such a $\gamma$ can be seen as summarizing all transitions in $M_\phi$ labeled by $A \subseteq \AP$ such that $A \sat_p \gamma$.
The construction given in Definition~\ref{def:mphi} suggests an immediate method for doing this:  rather than labeling transitions by $A \subseteq \AP$ such that $A \sat_p \bigwedge \textit{lits}(C_i)$, instead label a single transition by $\bigwedge \textit{lits}(C_i)$.
Representing transition labels symbolically in this manner also naturally allows multiple transitions to be grouped; the set of transitions from state $q$ to $q'$ can be combined into a singular ``edge'' in the automaton whose label is the boolean disjunction of the individual transition labels.  In some cases these disjunctions may in turn be reducible to more compact boolean formulas.

\underline{Relaxation of ANF.}
Our definition of ANF says that a formula is in ANF iff it has form $\bigvee C_i$, where each clause $C_i$ has form $(\bigwedge \mathcal{L}) \land \textbf{N}(\bigwedge \mathcal{F})$.  The method we give for converting formulas into ANF involves the use of a routine for converting formulas into DNF, which can itself be exponential.  We adopted this mechanism for ease of exposition, and also because in the worst case this exponential overhead is unavoidable.  However, requiring the propositional parts of clauses to be of form $\bigwedge \mathcal{L}$, where $\mathcal{L}$ consists only of literals, is unnecessarily restrictive:  all that is needed for the construction of $M_\phi$ is to require clauses to be of form $\gamma \land \mathbf{N} (\bigwedge \mathcal{F})$, where $\gamma \in \Gamma^\AP_e$ is a proposition formula in Extended Finite LTL.  Relaxing ANF in this manner eliminates the need for full DNF calculations in general, and can lead to time and space savings when  transitions are being represented symbolically.

\paragraph{Comparison with \ltlf Construction in~\cite{degiacomo:ijcai2013}}
We close this section with a brief comparison between the presented construction of $M_\phi$ and one for the logic \ltlf sketched in~\cite{degiacomo:ijcai2013}.
A primary focus in that paper was to consider the complexity of different decision problems related to \ltlf and another, more expressive logic, LDL$_f$, or Linear Dynamic Logic over finite traces.
Their results depend on automaton constructions, which they sketch in sufficient detail for the purposes of their complexity analyses.
In particular, the method they give for computing NFAs from \ltlf formulas relies first on a translation from \ltlf to LDL$_f$, then on a transformation from LDL$_f$ to alternating finite automata over words, then on known results in the literature for translating these alternating automata into traditional NFAs.
The number of states the NFA resulting from \ltlf formula $\phi$ is $2^{O(|\phi|)}$, with the invocation of the translations incurring constants that appear in the exponent.
In contrast, our construction yields a tight bound of $2^{|\phi|}$, with no constant in the exponent.
The syntactic link between \ltlf and NFA states also becomes obscured in their construction; in contrast we maintain an explicit link between NFA states and sets of Finite LTL formulas.  This link is especially important in our work on Finite LTL query checking~\cite{huang2020arxiv-qc}.

\section{Implementation and Empirical Results}
We have implemented our tableau construction as a C++ package.  The user specifies a formula $\phi \in \Phi^{\AP}_{e}$, and the corresponding NFA $M_{\phi}$ is constructed as defined in Section~\ref{sec:tableau}.
Our implementation uses the symbolic transition-label representation and on-the-fly construction of the set of states $Q_{\phi}$ discussed in~\ref{subsec:discussion}; transition labels are propositional formulas, and only states reachable from the initial state are added into $Q_\phi$.
Specifically, we maintain a set of states whose transitions have not  yet been computed; initially, this set contains only the (initial) state, which corresponds to exactly the set $\{\phi\}$.  We then repeatedly select a state from this set of states and compute the ANF representation of the conjunction of the formulas associated with that state.
Each of the disjuncts in this ANF representation defines a new (symbolic) transition for the automaton under construction.
These transitions and any new destination states are created, with all new destination states added to the set of states whose transitions have not yet been computed.
The process terminates when there are no states remaining that require the computation of their transitions.
Transitions between the same ordered pair of states are then collapsed into a single edge whose label is a disjunction of the transition labels of the individual transitions, which is then simplified.
Finally, the set of accepting states $F_{\phi}$ is determined syntactically through inspection of each state's corresponding formulas, using the method implied by Lemma~\ref{lem:empty-check}.

The Spot \cite{duret.16.atva2} platform (v.2.8.1) is used to handle the parsing of input formulas, and the Python package SymPy \cite{meurer:pcs2017} is employed to perform boolean simplification of the propositional edge labels.
To support proper parsing and construction of a Finite LTL formula, we added support for the $\tlWeakX$ (Weak Next) operator.
Additionally, several Spot-provided automatic formula rewrites are based on standard LTL identities  (such as $\tlX \logictrue \equiv \logictrue$) that do not hold under finite semantics; these were disabled.

\begin{figure}[!htp]
  \includegraphics[width=\textwidth]{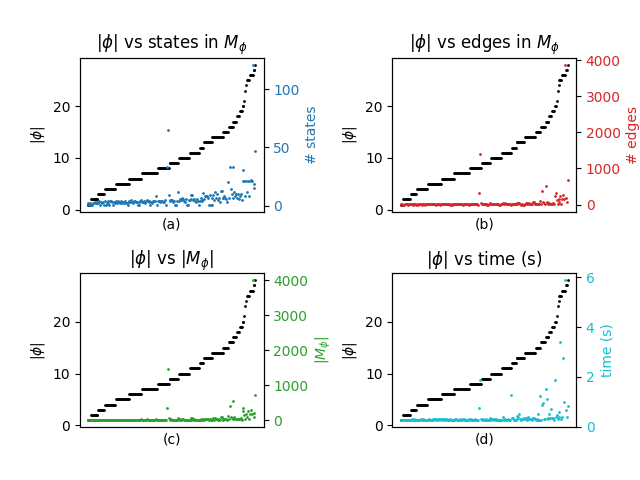}
  \caption{Summarized results of experiments.  Each plot contains a black dot for each of the formulas in the benchmark, which have been sorted in order of increasing formula complexity ($|\phi|$).  The colored dots, one per formula, in each of the four plots correspond to the following:
(a) states in automaton $M_{\phi}$;
  (b) edges in automaton $M_{\phi}$;
  (c) model size (sum of states and edges); and
  (d) computation time.
  }
  \label{fig:results-charts}
\end{figure}

To evaluate the implementation's performance we applied it on the benchmark set of 184 traditional LTL formulas (92 formulas and their negations) used by Duret-Lutz \cite{duret.14.ijccbs} to assess tools intended to construct automaton from traditional LTL.
As our semantics for Finite LTL differs from traditional LTL (and so too does the type of automaton that is produced), a direct comparison of our method to results of earlier uses of the benchmark for traditional LTL is not appropriate, so none is given.
For each formula of the benchmark, we calculated the complexity $|\phi|$ (the number of subformulas in $\phi$), the time required to perform the construction, and the number of states and edges in the resulting NFA.
Experiments were carried out on a single machine with an Intel Core i5-6600K (4 cores), with 32 GB RAM and a 64-bit version of GNU/Linux. 
A summarized set of the results are shown in Figure~\ref{fig:results-charts}, with formulas ordered by complexity.
Tables~\ref{tab:experimental-1} and~\ref{tab:experimental-2} in the appendix contain raw performance data for each formula of the benchmark.

Performing the construction generally completed in under one second for most formulas in the benchmark; the remainder completed in less than five seconds apiece.
As expected, formula complexity was found to be positively correlated with the size (states and transitions) of the constructed NFA.

%\begin{figure}
\begin{toappendix}
\section{Experimental Results}

The following tables contain full performance data for the experiments described in the paper.
\begin{table}
  \small
\caption{Experimental results:  Formulas 0--91.}\label{tab:experimental-1}
\begin{minipage}{0.45\textwidth}
\begin{tabular}{|*{5}{r|}}
%  \hline
%  \multicolumn{2}{|c|}{} & \multicolumn{3}{|c|}{Our approach} \\
  \hline
  ID & $|\phi|$ & states & edges & time(s) \\
  \hline 
  0 & 2 & 1 & 1 & 0.273 \\
  1 & 3 & 2 & 3 & 0.265 \\
  2 & 5 & 4 & 7 & 0.274 \\
  3 & 6 & 4 & 7 & 0.311 \\
  4 & 5 & 2 & 3 & 0.302 \\
  5 & 6 & 3 & 6 & 0.302 \\
  6 & 7 & 3 & 6 & 0.275 \\
  7 & 8 & 4 & 8 & 0.271 \\
  8 & 8 & 4 & 11 & 0.323 \\
  9 & 9 & 5 & 16 & 0.277 \\
  10 & 1 & 2 & 3 & 0.273 \\
  11 & 2 & 1 & 1 & 0.267 \\
  12 & 7 & 3 & 5 & 0.301 \\
  13 & 8 & 4 & 11 & 0.278 \\
  14 & 6 & 5 & 11 & 0.303 \\
  15 & 7 & 4 & 9 & 0.306 \\
  16 & 10 & 4 & 11 & 0.286 \\
  17 & 11 & 5 & 13 & 0.281 \\
  18 & 7 & 2 & 4 & 0.273 \\
  19 & 8 & 3 & 5 & 0.283 \\
  20 & 23 & 8 & 29 & 0.294 \\
  21 & 24 & 21 & 149 & 0.384 \\
  22 & 25 & 12 & 76 & 0.572 \\
  23 & 26 & 21 & 150 & 0.392 \\
  24 & 27 & 19 & 187 & 0.655 \\
  25 & 28 & 47 & 672 & 0.820 \\
  26 & 1 & 1 & 0 & 0.262 \\
  27 & 2 & 2 & 3 & 0.265 \\
  28 & 4 & 4 & 7 & 0.270 \\
  29 & 5 & 4 & 7 & 0.274 \\
  30 & 4 & 1 & 1 & 0.267 \\
  31 & 5 & 3 & 6 & 0.270 \\
  32 & 6 & 3 & 6 & 0.280 \\
  33 & 7 & 4 & 8 & 0.273 \\
  34 & 6 & 2 & 4 & 0.317 \\
  35 & 7 & 5 & 16 & 0.278 \\
  36 & 5 & 3 & 5 & 0.270 \\
  37 & 6 & 4 & 11 & 0.272 \\
  38 & 6 & 4 & 7 & 0.306 \\
  39 & 7 & 4 & 7 & 0.304 \\
  40 & 10 & 6 & 15 & 0.305 \\
  41 & 11 & 9 & 56 & 0.388 \\
  42 & 8 & 4 & 11 & 0.290 \\
  43 & 9 & 4 & 8 & 0.311 \\
  44 & 8 & 4 & 11 & 0.308 \\
  45 & 9 & 5 & 16 & 0.287 \\
  \hline
\end{tabular}
\end{minipage}
  \hspace{1em}
\begin{minipage}{0.45\textwidth}
\begin{tabular}{|*{5}{r|}}
%  \hline
%  \multicolumn{2}{|c|}{} & \multicolumn{3}{|c|}{Our approach} \\
  \hline
  ID & $|\phi|$ & states & edges & time(s) \\
  \hline 
  46 & 4 & 2 & 4 & 0.272 \\
  47 & 5 & 2 & 3 & 0.304 \\
  48 & 10 & 5 & 11 & 0.311 \\
  49 & 11 & 6 & 13 & 0.288 \\
  50 & 7 & 3 & 7 & 0.288 \\
  51 & 8 & 3 & 6 & 0.273 \\
  52 & 12 & 4 & 10 & 0.335 \\
  53 & 13 & 6 & 14 & 0.316 \\
  54 & 18 & 7 & 34 & 0.514 \\
  55 & 19 & 8 & 27 & 0.341 \\
  56 & 10 & 5 & 9 & 0.273 \\
  57 & 11 & 9 & 41 & 0.309 \\
  58 & 11 & 5 & 11 & 0.269 \\
  59 & 12 & 7 & 21 & 0.342 \\
  60 & 16 & 7 & 18 & 0.279 \\
  61 & 17 & 33 & 515 & 1.500 \\
  62 & 13 & 8 & 36 & 0.428 \\
  63 & 14 & 7 & 22 & 0.304 \\
  64 & 13 & 8 & 40 & 0.398 \\
  65 & 14 & 9 & 42 & 0.353 \\
  66 & 9 & 5 & 10 & 0.277 \\
  67 & 10 & 7 & 19 & 0.321 \\
  68 & 11 & 6 & 15 & 0.288 \\
  69 & 12 & 5 & 9 & 0.279 \\
  70 & 15 & 7 & 18 & 0.280 \\
  71 & 16 & 33 & 365 & 0.893 \\
  72 & 13 & 8 & 31 & 0.494 \\
  73 & 14 & 5 & 10 & 0.275 \\
  74 & 17 & 12 & 76 & 1.116 \\
  75 & 18 & 10 & 30 & 0.334 \\
  76 & 10 & 6 & 24 & 0.329 \\
  77 & 11 & 4 & 8 & 0.270 \\
  78 & 14 & 12 & 58 & 0.339 \\
  79 & 15 & 8 & 24 & 0.295 \\
  80 & 14 & 7 & 28 & 0.394 \\
  81 & 15 & 7 & 19 & 0.283 \\
  82 & 16 & 14 & 82 & 0.962 \\
  83 & 17 & 8 & 25 & 0.305 \\
  84 & 26 & 22 & 268 & 2.777 \\
  85 & 27 & 15 & 66 & 0.402 \\
  86 & 7 & 4 & 13 & 0.288 \\
  87 & 8 & 3 & 6 & 0.281 \\
  88 & 13 & 7 & 22 & 0.297 \\
  89 & 14 & 10 & 39 & 0.307 \\
  90 & 9 & 3 & 7 & 0.284 \\
  91 & 10 & 4 & 9 & 0.276 \\
  \hline
\end{tabular}
\end{minipage}
\end{table}

\begin{table}
  \small
\caption{Experimental results:  Formulas 92--183.}\label{tab:experimental-2}
\begin{minipage}{0.45\textwidth}
\begin{tabular}{|*{5}{r|}}
%  \hline
%  \multicolumn{2}{|c|}{} & \multicolumn{3}{|c|}{Our approach} \\
  \hline
  ID & $|\phi|$ & states & edges & time(s) \\
  \hline 
  92 & 15 & 6 & 22 & 0.515 \\
  93 & 16 & 10 & 40 & 0.314 \\
  94 & 20 & 21 & 236 & 1.884 \\
  95 & 21 & 21 & 117 & 0.463 \\
  96 & 9 & 4 & 13 & 0.292 \\
  97 & 10 & 3 & 7 & 0.284 \\
  98 & 16 & 7 & 22 & 0.322 \\
  99 & 17 & 10 & 39 & 0.357 \\
  100 & 11 & 3 & 7 & 0.298 \\
  101 & 12 & 4 & 10 & 0.278 \\
  102 & 18 & 6 & 22 & 0.725 \\
  103 & 19 & 10 & 40 & 0.365 \\
  104 & 25 & 21 & 236 & 3.390 \\
  105 & 26 & 21 & 162 & 0.981 \\
  106 & 19 & 5 & 15 & 0.274 \\
  107 & 20 & 31 & 315 & 0.395 \\
  108 & 25 & 8 & 29 & 0.294 \\
  109 & 26 & 121 & 3875 & 5.891 \\
  110 & 1 & 2 & 3 & 0.268 \\
  111 & 2 & 2 & 3 & 0.282 \\
  112 & 2 & 3 & 6 & 0.263 \\
  113 & 3 & 4 & 11 & 0.312 \\
  114 & 5 & 4 & 11 & 0.275 \\
  115 & 6 & 3 & 6 & 0.272 \\
  116 & 6 & 3 & 5 & 0.270 \\
  117 & 7 & 1 & 0 & 0.310 \\
  118 & 3 & 1 & 0 & 0.275 \\
  119 & 4 & 4 & 11 & 0.267 \\
  120 & 2 & 2 & 2 & 0.272 \\
  121 & 3 & 3 & 6 & 0.277 \\
  122 & 9 & 1 & 0 & 0.269 \\
  123 & 10 & 6 & 13 & 0.272 \\
  124 & 6 & 1 & 0 & 0.310 \\
  125 & 7 & 3 & 5 & 0.307 \\
  126 & 13 & 1 & 0 & 0.283 \\
  127 & 14 & 5 & 13 & 0.294 \\
  128 & 2 & 2 & 3 & 0.277 \\
  129 & 3 & 2 & 3 & 0.302 \\
  130 & 8 & 5 & 9 & 0.272 \\
  131 & 9 & 1 & 0 & 0.273 \\
  132 & 9 & 4 & 8 & 0.279 \\
  133 & 10 & 1 & 0 & 0.268 \\
  134 & 14 & 7 & 26 & 0.293 \\
  135 & 15 & 3 & 5 & 0.273 \\
  136 & 13 & 9 & 28 & 0.321 \\
  137 & 14 & 3 & 5 & 0.275 \\
  \hline
\end{tabular}
\end{minipage}
  \hspace{1em}
\begin{minipage}{0.45\textwidth}
\begin{tabular}{|*{5}{r|}}
%  \hline
%  \multicolumn{2}{|c|}{} & \multicolumn{3}{|c|}{Our approach} \\
  \hline
  ID & $|\phi|$ & states & edges & time(s) \\
  \hline 
  138 & 4 & 2 & 4 & 0.307 \\
  139 & 5 & 2 & 3 & 0.299 \\
  140 & 6 & 3 & 5 & 0.271 \\
  141 & 7 & 2 & 4 & 0.277 \\
  142 & 7 & 4 & 12 & 0.273 \\
  143 & 8 & 1 & 0 & 0.277 \\
  144 & 4 & 1 & 0 & 0.288 \\
  145 & 5 & 3 & 5 & 0.305 \\
  146 & 4 & 4 & 8 & 0.302 \\
  147 & 5 & 2 & 2 & 0.267 \\
  148 & 8 & 33 & 312 & 0.753 \\
  149 & 9 & 12 & 52 & 0.347 \\
  150 & 13 & 1 & 0 & 0.323 \\
  151 & 14 & 13 & 66 & 0.357 \\
  152 & 13 & 1 & 0 & 0.349 \\
  153 & 14 & 13 & 75 & 0.366 \\
  154 & 10 & 1 & 0 & 0.274 \\
  155 & 11 & 6 & 13 & 0.272 \\
  156 & 6 & 1 & 0 & 0.266 \\
  157 & 7 & 4 & 6 & 0.270 \\
  158 & 2 & 3 & 5 & 0.274 \\
  159 & 3 & 3 & 5 & 0.276 \\
  160 & 3 & 1 & 0 & 0.268 \\
  161 & 4 & 4 & 11 & 0.277 \\
  162 & 4 & 3 & 5 & 0.269 \\
  163 & 5 & 4 & 11 & 0.318 \\
  164 & 14 & 7 & 13 & 0.305 \\
  165 & 15 & 20 & 90 & 1.237 \\
  166 & 4 & 1 & 0 & 0.267 \\
  167 & 5 & 2 & 4 & 0.271 \\
  168 & 6 & 4 & 6 & 0.273 \\
  169 & 7 & 4 & 10 & 0.278 \\
  170 & 4 & 3 & 5 & 0.266 \\
  171 & 5 & 2 & 4 & 0.310 \\
  172 & 5 & 1 & 0 & 0.315 \\
  173 & 6 & 5 & 20 & 0.276 \\
  174 & 4 & 2 & 4 & 0.273 \\
  175 & 5 & 3 & 6 & 0.275 \\
  176 & 10 & 5 & 9 & 0.272 \\
  177 & 11 & 4 & 7 & 0.287 \\
  178 & 11 & 1 & 0 & 1.261 \\
  179 & 12 & 11 & 25 & 0.275 \\
  180 & 7 & 8 & 23 & 0.275 \\
  181 & 8 & 65 & 1395 & 1.866 \\
  182 & 7 & 1 & 1 & 0.305 \\
  183 & 8 & 9 & 50 & 0.311 \\
  \hline
\end{tabular}
\end{minipage}
\end{table}
\end{toappendix}
%\en{figure}

\section{Conclusion}

This paper has given a tableau-based method for constructing nondeterministic finite automata (NFAs) from formulas in a version of Linear Temporal Logic (LTL) whose formulas are intepreted with respect to finite, rather than infinite, sequences of states.  It first introduces the logic, Finite LTL, under consideration, defines its syntax and semantics, and establishes that it is strictly more expressive than the logic \ltlf~\cite{degiacomo:ijcai2013}.  This latter result is due to the fact, that in contrast to \ltlf, the empty sequence is allowed in the semantics of Finite LTL.  The paper then gives a series of syntactic transformations that are used to rewrite Finite LTL formulas in \emph{automaton normal form} (ANF).  The states and transitions of the NFA of a given Finite LTL formula are extracted from the ANF equivalent of the formula, with each state being associated with a set of subformulas of the original formula.  We show that the resulting NFA has at most $2^{|\phi|}$ states, where $\phi$ is the formula from which the automaton is constructed.  The description of an implementation is then given, and experimental results reported on a benchmark of 184 formulas given in the literature.  Although our implementation is lightly optimized, each formula's automaton is computed in fewer than five seconds, with most taking less than a second.

As ongoing work, we have used this construction as a basis for \emph{Finite LTL query checking}~\cite{huang2020arxiv-qc}.  In Finite LTL query checking, one is given a finite set of finite sequences, and a Finite LTL query, or formula with a missing subformula; the goal of query checking is to compute all of the solutions for the missing subformula that make the resulting completed formula true for all the given finite sequences.  The motivation for such query-checking is to mine temporal properties based on the queries from the finite-length sequences; these can then be used for system understanding and other analysis tasks.  Our approach relies heavily on the construction in this paper to compute NFAs, as it needs access to the formulas / queries associated with a given automaton state.

For future work, we would like to extend our construction to other linear-time logics.  In particular, the linear-time mu-calculus~\cite{vardi1988temporal} over finite sequences is a natural candidate to consider due to its expressiveness and associated ability to uniformly encode a wide range of other linear-time logics, including the logic LDL$_f$ of~\cite{degiacomo:ijcai2013}.

%This paper has defined the logic Finite LTL, which uses the syntax of LTL but employs a semantics based on finite, rather than infinite, state sequences.  It also has provided a tableau-inspired construction for converting formulas into non-deterministic finite automata whose languages consist of exactly the sequences making the corresponding formulas true.  In the construction states are equivalent to subsets of subformulas of the input formula, and accepting states are defined as those where all contained subformulas satisfy the empty sequence.  We have also observed that the check for satisfaction by the empty sequence can be done purely syntactically on the basis of the structure of the formula.  We also have described a prototype implementation of the approach and given empirical results on an existing benchmark for standard LTL.  Our methodology, while heavily inspired by the standard (infinite semantics) LTL community, transforms a Finite LTL formula directly into an NFA, and does not rely upon any intermediate representation from the standard LTL realm.  The continued development of such ``native'' approaches will allow us to appropriately and accurately reason over domains with finite sequences.

\bibliography{local}

\begin{thebibliography}{10}

\bibitem{alur1993model}
Rajeev Alur, Costas Courcoubetis, and David Dill.
\newblock Model-checking in dense real-time.
\newblock {\em Information and Computation}, 104(1):2--34, 1993.

\bibitem{aziz1995usually}
Adnan Aziz, Vigyan Singhal, Felice Balarin, Robert~K Brayton, and Alberto~L
  Sangiovanni-Vincentelli.
\newblock It usually works: The temporal logic of stochastic systems.
\newblock In {\em International Conference on Computer-Aided Verification},
  pages 155--165. Springer, 1995.

\bibitem{bauer2010comparing}
Andreas Bauer, Martin Leucker, and Christian Schallhart.
\newblock Comparing {LTL} semantics for runtime verification.
\newblock {\em Journal of Logic and Computation}, 20(3):651--674, 2010.

\bibitem{chan:cav2000}
William Chan.
\newblock Temporal-logic queries.
\newblock In E.~Allen Emerson and A.~Prasad Sistla, editors, {\em International
  Conference on Computer-Aided Verification}, volume 1855 of {\em Lecture Notes
  in Computer Science}, pages 450--463. Springer, 2000.

\bibitem{ashok1981alternation}
Ashok~K. Chandra, Dexter~C. Kozen, and Larry~J. Stockmeyer.
\newblock Alternation.
\newblock {\em Journal of the ACM}, 28(1):114--133, January 1981.

\bibitem{chockler:hsvt2011}
Hana Chockler, Arie Gurfinkel, and Ofer Strichman.
\newblock Variants of {LTL} query checking.
\newblock In Sharon Barner, Ian Harris, Daniel Kroening, and Orna Raz, editors,
  {\em Hardware and Software: Verification and Testing}, pages 76--92, Berlin,
  Heidelberg, 2011. Springer Berlin Heidelberg.

\bibitem{courveur:fm1999}
Jean-Michel Couvreur.
\newblock On-the-fly verification of linear temporal logic.
\newblock In Jeannette~M. Wing, Jim Woodcock, and Jim Davies, editors, {\em
  FM'99 --- Formal Methods}, pages 253--271. Springer Verlag, 1999.

\bibitem{degiacomo:aaai2014}
Giuseppe De~Giacomo, Riccardo De~Masellis, and Marco Montali.
\newblock Reasoning on {LTL} on finite traces: Insensitivity to infiniteness.
\newblock In {\em AAAI Conference on Artificial Intelligence}, AAAI'14, pages
  1027--1033. AAAI Press, 2014.

\bibitem{degiacomo:ijcai2013}
Giuseppe De~Giacomo and Moshe~Y. Vardi.
\newblock Linear temporal logic and linear dynamic logic on finite traces.
\newblock In {\em International Joint Conference on Artificial Intelligence},
  IJCAI '13, pages 854--860. AAAI Press, 2013.

\bibitem{duret.14.ijccbs}
Alexandre Duret-Lutz.
\newblock {LTL} translation improvements in {S}pot 1.0.
\newblock {\em International Journal on Critical Computer-Based Systems},
  5(1/2):31--54, March 2014.

\bibitem{duret.16.atva2}
Alexandre Duret-Lutz, Alexandre Lewkowicz, Amaury Fauchille, Thibaud Michaud,
  Etienne Renault, and Laurent Xu.
\newblock Spot 2.0 --- a framework for {LTL} and $\omega$-automata
  manipulation.
\newblock In {\em International Symposium on Automated Technology for
  Verification and Analysis}, volume 9938 of {\em Lecture Notes in Computer
  Science}, pages 122--129. Springer, October 2016.

\bibitem{emerson1986sometimes}
E.~Allen Emerson and Joseph~Y. Halpern.
\newblock ``{S}ometimes'' and ``not never'' revisited: On branching versus
  linear time temporal logic.
\newblock {\em Journal of the ACM}, 33(1):151--178, 1986.

\bibitem{fionda:aaai2016}
Valeria Fionda and Gianluigi Greco.
\newblock The complexity of {LTL} on finite traces: Hard and easy fragments.
\newblock In {\em AAAI Conference on Artificial Intelligence}, AAAI'16, pages
  971--977. AAAI Press, 2016.

\bibitem{gerevini2009deterministic}
Alfonso~E Gerevini, Patrik Haslum, Derek Long, Alessandro Saetti, and Yannis
  Dimopoulos.
\newblock Deterministic planning in the {Fifth} {International} {Planning}
  {Competition}: {PDDL3} and experimental evaluation of the planners.
\newblock {\em Artificial Intelligence}, 173(5-6):619--668, 2009.

\bibitem{huang:avocs2017}
Samuel Huang and Rance Cleaveland.
\newblock Query checking for linear temporal logic.
\newblock In Laure Petrucci, Cristina Seceleanu, and Ana Cavalcanti, editors,
  {\em Critical Systems: Formal Methods and Automated Verification}, pages
  34--48, Cham, 2017. Springer International Publishing.

\bibitem{huang2020arxiv-qc}
Samuel Huang and Rance Cleaveland.
\newblock Temporal-logic query checking over finite data streams, 2020.
\newblock arXiv:2006.03751, submitted for publication.

\bibitem{kozen1983results}
Dexter Kozen.
\newblock Results on the propositional $\mu$-calculus.
\newblock {\em Theoretical Computer Science}, 27(3):333--354, 1983.

\bibitem{kress2009temporal}
H.~{Kress-Gazit}, G.~E. {Fainekos}, and G.~J. {Pappas}.
\newblock Temporal-logic-based reactive mission and motion planning.
\newblock {\em IEEE Transactions on Robotics}, 25(6):1370--1381, 2009.

\bibitem{li:cav2019}
Jianwen Li, Moshe~Y. Vardi, and Kristin~Y. Rozier.
\newblock Satisfiability checking for mission-time {LTL}.
\newblock In Isil Dillig and Serdar Tasiran, editors, {\em International
  Conference on Computer-Aided Verification}, pages 3--22. Springer
  International Publishing, 2019.

\bibitem{li:arxiv2014}
Jianwen Li, Lijun Zhang, Geguang Pu, Moshe~Y. Vardi, and Jifeng He.
\newblock {LTL}$_f$ satisfiability checking.
\newblock {\em CoRR}, abs/1403.1666, 2014.

\bibitem{maggi2011runtime}
Fabrizio~Maria Maggi, Michael Westergaard, Marco Montali, and Wil~M.P. van~der
  Aalst.
\newblock Runtime verification of {LTL}-based declarative process models.
\newblock In {\em International Conference on Runtime Verification}, pages
  131--146. Springer, 2011.

\bibitem{meurer:pcs2017}
Aaron Meurer, Christopher~P. Smith, Mateusz Paprocki, et~al.
\newblock {SymPy}: Symbolic computing in {Python}.
\newblock {\em PeerJ Computer Science}, 3:e103, January 2017.

\bibitem{pnueli:sfcs77}
Amir {Pnueli}.
\newblock The temporal logic of programs.
\newblock In {\em Symposium on Foundations of Computer Science}, pages 46--57,
  October 1977.

\bibitem{rosu:rv2016}
Grigore Ro{\c{s}}u.
\newblock Finite-trace linear temporal logic: Coinductive completeness.
\newblock In Yli{\`e}s Falcone and C{\'e}sar S{\'a}nchez, editors, {\em Runtime
  Verification}, pages 333--350. Springer International Publishing, 2016.

\bibitem{vardi1988temporal}
Moshe~Y. Vardi.
\newblock A temporal fixpoint calculus.
\newblock In {\em Symposium on Principles of Programming Languages}, pages
  250--259, 1988.

\bibitem{vardi:automata1986}
Moshe~Y. Vardi and Pierre Wolper.
\newblock An automata-theoretic approach to automatic program verification.
\newblock In {\em Symposium on Logic in Computer Science}, pages 322--331. IEEE
  Computer Society, 1986.

\bibitem{wolper1985tableau}
Pierre Wolper.
\newblock The tableau method for temporal logic: An overview.
\newblock {\em Logique et Analyse}, pages 119--136, 1985.

\end{thebibliography}
\end{document}